\RequirePackage{etex}
\documentclass[11pt, a4paper]{article}
\usepackage[utf8]{inputenc}
\usepackage{inputenc}
\usepackage{amsmath}
\usepackage{amssymb}
\usepackage{stmaryrd} 
\usepackage{upgreek}
\usepackage{graphicx}
\usepackage{authblk}
\usepackage{layout}
\usepackage{dsfont}
\usepackage{bbold}
\usepackage[english]{babel}
\usepackage[top=2.2cm, bottom=3.2cm, left=2.2cm, right=2.2cm]{geometry}
\usepackage{multicol}
\usepackage{amsthm}
\usepackage{mathrsfs}
\usepackage{tikz}
\usepackage{tikz-network}
\usepackage[toc,page]{appendix}
\usepackage{hyperref}
\usepackage{complexity}
\usepackage[margin=1cm]{caption}
\usepackage{enumitem}
\usepackage{fontawesome}
\usepackage{todonotes}
\usepackage{xspace}
\usepackage{subcaption}

\def\rotateclockwise#1{
  \newdimen\xrw
  \pgfextractx{\xrw}{#1}
  \newdimen\yrw
  \pgfextracty{\yrw}{#1}
  \pgfpoint{\yrw}{-\xrw}
}

\def\rotatecounterclockwise#1{
  \newdimen\xrcw
  \pgfextractx{\xrcw}{#1}
  \newdimen\yrcw
  \pgfextracty{\yrcw}{#1}
  \pgfpoint{-\yrcw}{\xrcw}
}

\def\outsidespacerpgfclockwise#1#2#3{
  \pgfpointscale{#3}{
    \rotateclockwise{
      \pgfpointnormalised{
        \pgfpointdiff{#1}{#2}}}}
}

\def\outsidespacerpgfcounterclockwise#1#2#3{
  \pgfpointscale{#3}{
    \rotatecounterclockwise{
      \pgfpointnormalised{
        \pgfpointdiff{#1}{#2}}}}
}

\def\outsidepgfclockwise#1#2#3{
  \pgfpointadd{#2}{\outsidespacerpgfclockwise{#1}{#2}{#3}}
}

\def\outsidepgfcounterclockwise#1#2#3{
  \pgfpointadd{#2}{\outsidespacerpgfcounterclockwise{#1}{#2}{#3}}
}

\def\outside#1#2#3{
  ($ (#2) ! #3 ! -90 : (#1) $)
}

\def\cornerpgf#1#2#3#4{
  \pgfextra{
    \pgfmathanglebetweenpoints{#2}{\outsidepgfcounterclockwise{#1}{#2}{#4}}
    \let\anglea\pgfmathresult
    \let\startangle\pgfmathresult

    \pgfmathanglebetweenpoints{#2}{\outsidepgfclockwise{#3}{#2}{#4}}
    \pgfmathparse{\pgfmathresult - \anglea}
    \pgfmathroundto{\pgfmathresult}
    \let\arcangle\pgfmathresult
    \ifthenelse{180=\arcangle \or 180<\arcangle}{
      \pgfmathparse{-360 + \arcangle}}{
      \pgfmathparse{\arcangle}}
    \let\deltaangle\pgfmathresult

    \newdimen\x
    \pgfextractx{\x}{\outsidepgfcounterclockwise{#1}{#2}{#4}}
    \newdimen\y
    \pgfextracty{\y}{\outsidepgfcounterclockwise{#1}{#2}{#4}}
  }
  -- (\x,\y) arc [start angle=\startangle, delta angle=\deltaangle, radius=#4]
}

\def\corner#1#2#3#4{
  \cornerpgf{\pgfpointanchor{#1}{center}}{\pgfpointanchor{#2}{center}}{\pgfpointanchor{#3}{center}}{#4}
}

\def\hedgem#1#2#3#4{
  
  \outside{#1}{#2}{#4}
  \pgfextra{
    \def\hgnodea{#1}
    \def\hgnodeb{#2}
  }
  foreach \c in {#3} {
    \corner{\hgnodea}{\hgnodeb}{\c}{#4}
    \pgfextra{
      \global\let\hgnodea\hgnodeb
      \global\let\hgnodeb\c
    }
  }
  \corner{\hgnodea}{\hgnodeb}{#1}{#4}
  \corner{\hgnodeb}{#1}{#2}{#4}
  -- cycle
}

\tikzstyle{noeud}=[circle, inner sep=1.5, minimum size =2 pt, line width = .8pt, draw=black, fill=white]
\tikzstyle{R}=[circle, inner sep=1.5, minimum size =6 pt, line width = 1.5pt, draw=red, fill=red]
\tikzstyle{B}=[circle, inner sep=1.5, minimum size =6 pt, line width = 1.5pt, draw=blue, fill=blue]
\tikzstyle{inv}=[circle,inner sep=0, minimum size =4 pt, line width = 1pt, draw=white, fill=white, text= black]

 \tikzstyle{rond}=[circle, draw = green, inner sep=2mm, line width = 1.5pt]
 \tikzstyle{croix}=[cross out, draw = blue, inner sep=3mm]

\tikzstyle{v}=[circle,inner sep=0, minimum size =3 pt, line width = 1pt, draw=black, fill=black, text= white]

\tikzstyle{ghost}=[circle,inner sep=0, minimum size =4 pt, line width = 1pt, draw=black!20, fill=black!20, text= white]

\tikzstyle{arete}=[draw, line width=1.5pt, draw=black]

\tikzstyle{decision} = [diamond, draw, fill=yellow!20, 
    text width=10em, text centered, node distance=5cm, inner sep=0pt]
\tikzstyle{block} = [rectangle, draw, fill=yellow!20, 
    text width=10em, text centered, rounded corners, minimum height=4em]
\tikzstyle{line} = [draw, -latex']
    
\tikzstyle{sortieD} = [draw, ellipse,fill=blue!20]
 \tikzstyle{sortieA} = [draw, ellipse,fill=red!20]

\newtheorem{theorem}{Theorem}
\newtheorem{corollary}[theorem]{Corollary}
\newtheorem{proposition}[theorem]{Proposition}
\newtheorem{lemma}[theorem]{Lemma}

\newtheoremstyle{cla}{13pt}{12pt}{\it}{}{$\quad$}{.}{ }{}
{\theoremstyle{cla} \newtheorem{claim}{Claim}}

\newenvironment{proofclaim}[1][Proof]{\begin{proof}[#1]}{\end{proof}}

\newcommand{\EXPTIME}{\lang{EXPTIME}\xspace}
\newcommand{\hyp}{\mathcal{H}}
\newcommand{\strat}{\mathcal{S}}
\newcommand{\pos}{\mathcal{P}}
\newcommand{\Zp}{\mathbb{Z}_+}
\newcommand{\interval}[2]{\left\llbracket #1,#2 \right\rrbracket}
\newcommand{\ov}[1]{\overline{#1}}

\title{Token positional games\thanks{This research was supported by the ANR project P-GASE (ANR-21-CE48-0001-01)}}

\author[1]{Guillaume Bagan} 
\author[1]{Quentin Deschamps}
\author[2]{Florian Galliot}
\author[3]{Mirjana Mikala\v{c}ki \thanks{Partly supported by Provincial Secretariat for Higher Education and Scientific Research, Province of Vojvodina (Grant No.~142-451-2686/2021). Partly supported by Ministry of Science,
Technological Development and Innovation of Republic of Serbia
(Grants 451-03-66/2024-03/200125 \& 451-03-65/2024-03/200125).
}}
\author[4]{Nacim Oijid\thanks{Kempe Foundation Grant No. JCSMK24-515 (Sweden)}}

\affil[1]{Univ Lyon, CNRS, UCBL, INSA Lyon, LIRIS, UMR5205, F-69622 Villeurbanne, France}
\affil[2]{Aix-Marseille Université, CNRS, Centrale Marseille, I2M, UMR 7373, 13453 Marseille, France}
\affil[3]{Department of Mathematics and Informatics, Faculty of Sciences, University of Novi Sad, Serbia}
\affil[4]{Umeå University, Sweden}
\date{}

\begin{document}

\maketitle

\begin{abstract}
    The classical Maker-Breaker positional game is played on a board which is a hypergraph $\hyp$, with two players, Maker and Breaker, alternately claiming vertices of $\hyp$ until all the vertices are claimed. When the game ends, Maker wins if she has claimed all the vertices of some edge of $\hyp$; otherwise, Breaker wins. Playing this game in real life can be done by placing tokens on the vertices of the board. In this paper, we study the unfortunate case in which one or both players do not have enough tokens to cover all the vertices and, as such, will have to move their tokens around at some point instead of placing new ones. There may be a bias, in that Maker and Breaker do not necessarily have the same amount of tokens. The present paper initiates the study of this generalization of positional games, called {\em token positional games}.
    
    A particularly interesting case is when Maker has a winning strategy in the classical game: what is the lowest number of tokens with which she still wins against Breaker's unlimited stock? We notably show that, for $k$-uniform hypergraphs on an arbitrarily large number $n$ of vertices, this number equals $k$ if $k \in\{2,3\}$ but can vary from $k$ to $\Omega(n)$ if $k \geq 4$. From an algorithmic point of view, {\sf PSPACE}-hardness in general is inherited from classical positional games, but we get a polynomial-time algorithm to solve the case where Breaker only has one token. We also establish {\sf EXPTIME}-completeness for a ``token sliding'' variation of the game.
\end{abstract}

\section{Introduction}

Classical {\em Maker-Breaker} games \cite{chvatal1978} are two-player combinatorial games played on hypergraphs. Let $\hyp$ be a finite hypergraph whose vertex set $V(\hyp)$ we call the {\em board} of the game and whose edge set $E(\hyp)$ may be interpreted as a family of {\em winning sets}. The two players, called Maker and Breaker, take turns in claiming previously unclaimed elements of the board, with Maker going first. Maker wins the game if, by the end of the game, she has claimed all the elements of some winning set: we say that she has {\em filled} that winning set. Otherwise, Breaker is declared the winner of the game. Note that no draw is possible. We define the {\em outcome} as the function that maps a hypergraph $\hyp$ to the player that has a winning strategy on $\hyp$ (for short, we say that this player {\em wins on $\hyp$}).

Due to their convenient properties, Maker-Breaker games have been by far the most studied among a larger family of games on hypergraphs called {\em positional games}. The general study of positional games started in the second half of the 20th century~\cite{Hales1963, erdos1973}, and developed over the years in roughly two directions. On the one hand, different {\em conventions} are being considered: this means either other winning conditions for the players, as in the {\em Maker-Maker} \cite{Hales1963}, {\em Avoider-Avoider} \cite{Har81} and {\em Avoider-Enforcer} \cite{Lu92} conventions, or alternative ways to select vertices, as in the {\em Client-Waiter} and {\em Waiter-Client} \cite{beckpicker} conventions. On the other hand, the general framework of positional games has been extended in several ways: in recent years, we have seen the introduction of scoring positional games, in which we count the number of edges filled by Maker instead of asking whether she can fill one or not~\cite{BaganIncidence}, as well as a version where the vertices are partially ordered, forcing some moves to be played before others~\cite{BAGANpopo}.

The study of all these variations of positional games usually aims, for various hypergraph classes, at characterizing hypergraphs on which such or such player wins, and determining the algorithmic complexity of computing the outcome of the game. The latter aspect is typically studied depending on the size of the edges. The {\em rank} of a hypergraph is defined as the size of its largest edge, and a hypergraph is deemed {\em $k$-uniform} if all its edges have size exactly $k$. Maker-Breaker games were first proved to be \PSPACE-complete on hypergraphs of rank 11 through the game \textsf{POS CNF}~\cite{schaefer1978complexity}, and this result has been improved several times during the last four years~\cite{rahman20236,koepke2025advances,galliot2025}, so it is now known that they remain \PSPACE-complete even restricted to 4-uniform hypergraphs. Under classical complexity assumptions, this bound is tight since Maker-Breaker games can be solved in polynomial time on hypergraphs of rank 3~\cite{GGS25}.

Historically, a lot of Maker-Breaker games were played on a complete graph of which the players were claiming edges, with Maker aiming at claiming a subgraph with some prescribed structure. Since Maker would usually win these games very easily, Breaker was given more power: a {\em bias} was introduced, allowing Breaker to claim $q>1$ elements per move, while Maker could still only claim one~\cite{chvatal1978}. This led to an abundant literature on the {\em threshold bias} of various Maker-Breaker games, i.e. the smallest value of $q$ such that Breaker wins~\cite{chvatal1978, beck1982, Bednarska2000}. Later, Maker was also given a bias $p$, leading to the study of the threshold bias of Breaker as a function of $p$~\cite{balogh2009diameter,beckbook2008,gebauer2012clique, hefetz2012doubly, fouadi2025asymptotic}.

In this paper, we introduce {\em token positional games} as a different way to change the balance of the Maker-Breaker game. Instead of biases, we give the players tokens: $a$ red tokens for Maker, and $b$ blue tokens for Breaker.
The {\em $(a,b)$-game} on a hypergraph $\hyp$ is then defined as follows. On their turn, each player places one of their own tokens on an unoccupied vertex. This might be done by placing a token that had not yet been used, or by moving one that was already placed on the board (note that, once all tokens are placed, it is necessary to play according to the latter case). We allow for pass moves. Note that, when a player moves a token, the vacated vertex becomes available for the other player (or for that same player) to place a token later. A player is allowed to move a token from its position before all of their tokens are already used, but this is suboptimal since it leaves that vertex available for the opponent. Maker wins in this setting as in the classical version, i.e.\ if she manages to fill some edge with red tokens, while Breaker wins if he manages to indefinitely prevent Maker from winning or if the game reaches the same state twice. Note that, if $\varnothing \in E(\hyp)$, then Maker wins before the game even starts. As in the classical Maker-Breaker game, Breaker wins if the set of blue tokens on board forms a transversal of $\hyp$, as he may simply pass all his upcoming moves and the game will reach the same state twice. We use the symbol $\ast$ when a player is given enough tokens to be able to use a new one on each turn, which we may interpret as that player having an unlimited amount of tokens. Note that the $(\ast,\ast)$-game coincides with the classical Maker-Breaker game. The ``token'' point of view is not relevant for players that have an unlimited amount, so we may simply say that they {\em claim} vertices as in the classical game.

This variation of positional games involving tokens was first mentioned in the PhD thesis of the third author~\cite{galliot2023hypergraphs}, but such games actually have a history of their own. Namely, it was mentioned in~\cite{zaslavsky1982tic} that in the ancient Roman Empire, there existed some version of Tic-Tac-Toe called {\em Terni Lapilli} where both players were given three tokens, which they would place on the board in the first three rounds and then move around until some player gets the desired three-in-line configuration (which never happens with optimal play). That game actually included some restrictions on the way the tokens could be moved, but a similar game without restrictions was also known in France under the name {\em Les Pendus}, as noted by Kraitchik and Gardner~\cite{gardner1959,kraitchik1942}. The game of {\em Nine Men’s Morris}, which likely dates back to the Roman Empire as well, is also of a similar nature although the game does not end when a player gets three aligned tokens: instead, that player gets to remove one of the opponent’s tokens from the board.

Section \ref{section2} features some first results on token positional games, which will be useful in later sections. Section \ref{section3} solves the case where Breaker has a single token: in particular, we get a polynomial-time algorithm. Section \ref{section4}, on the contrary, addresses the case where Breaker has unlimited tokens. We consider the threshold number of tokens that Maker needs to win depending on the size $k$ of the edges, which we relate to the duration of the classical game which is a well-studied parameter in positional games. We provide bounds on both parameters which are either exact or almost tight, for all edge sizes. Section \ref{section5} is dedicated to algorithmic results. Token positional games are \PSPACE-hard in general but lie in \XP~parameterized by the number of tokens of both players. We also get an \EXPTIME-completeness result for a version of the game where the tokens are restricted to be slid along the edges. Section \ref{section6} concludes the paper and suggests some leads for future research.

\section{Preliminaries and general results}\label{section2}

A trivial remark is that more tokens is always better.

\begin{proposition}
    Let $\hyp$ be a hypergraph, and let $a,a',b,b' \in \Zp$ such that $a' \geq a$ and $b' \leq b$. If Maker wins the $(a,b)$-game on $\hyp$, then Maker wins the $(a',b')$-game on $\hyp$.
\end{proposition}

\begin{proof}
    Suppose that Maker wins the $(a,b)$-game on $\hyp$. Maker can apply the same strategy to win the $(a',b)$-game on $\hyp$: she simply never uses the $a'-a$ extra tokens. Similarly, it is therefore impossible that Breaker wins the $(a',b')$-game on $\hyp$, since he could then apply the same strategy to win the $(a',b)$-game on $\hyp$.
\end{proof}

A {\em subhypergraph} of a hypergraph $\hyp$ is a hypergraph $\hyp'$ such that $V(\hyp') \subseteq V(\hyp)$ and $E(\hyp') \subseteq E(\hyp)$ (note that this need not be an induced structure). A convenient property of classical Maker-Breaker games, which largely explains why this convention is the most studied, is the fact that Maker winning on a subhypergraph $\hyp'$ of $\hyp$ implies that Maker also wins on $\hyp$, as she may ignore any action that takes place outside of $\hyp'$. This property still holds for token positional games.

\begin{proposition}[Subhypergraph monotonicity]
    Let $\hyp$ be a hypergraph, let $\hyp'$ be a subhypergraph of $\hyp$, and let $a,b \in \Zp$. If Maker wins the $(a,b)$-game on $\hyp'$, then Maker wins the $(a,b)$-game on $\hyp$.
\end{proposition}

\begin{proof}
    If Maker wins the $(a,b)$-game on $\hyp'$, then she may apply the same winning strategy on $\hyp$, playing all of her moves inside $V(\hyp')$ and eventually filling an edge in $E(\hyp') \subseteq E(\hyp)$. Indeed, if Breaker ever places a token outside $V(\hyp')$, Maker can pretend that Breaker has placed that token on some arbitrary unoccupied vertex in $V(\hyp')$ instead, and go on with her winning strategy.
\end{proof}

A {\em pairing} is a set of pairwise disjoint pairs of vertices. In classical Maker-Breaker games, any pairing $\Pi$ has an associated {\em pairing strategy} for Breaker. This means that, if Maker claims a vertex $x \in \{x,y\} \in \Pi$ where $y$ is available, then Breaker answers by claiming $y$, otherwise Breaker claims an arbitrary vertex. A pairing is said to be {\em complete} in a hypergraph $\hyp$ if every edge $e \in E(\hyp)$ is {\em covered} i.e. $\pi \subseteq e$ for some $\pi \in \Pi$. If some pairing $\Pi$ is complete in $\hyp$, then the pairing strategy associated to $\Pi$ is a winning Breaker strategy for the Maker-Breaker game on $\hyp$, as it ensures that he claims at least one vertex per edge.

Pairing strategies can obviously be adapted to token positional games, provided Breaker has enough tokens to carry them out. As we will also need such strategies mid-game, we introduce some definitions. A {\em position} of the $(a,b)$-game is a triple $(\hyp,M,B)$ where $\hyp$ is a hypergraph and $M,B \subseteq V(\hyp)$ correspond to the current tokens that the respective players have on the board: $M \cap B = \varnothing$, $|M| \leq a$ and $|B| \leq b$. Note that the position $(\hyp,\varnothing,\varnothing)$ is simply the starting position when playing the game on the hypergraph $\hyp$. We can define the vertex set and the edge set of a position $\pos=(\hyp,M,B)$ as $V(\pos)=V(\hyp) \setminus (M \cup B)$ and $E(\pos)=\{e \setminus M \mid e \in E(\hyp), e \cap B = \varnothing\}$ respectively. This is a natural definition since, in the position $(\hyp,M,B)$, an edge $e \in E(\hyp)$ is either disjoint from $B$ in which case $e \setminus M$ is all that is left to fill for Maker, or it intersects $B$ in which case Breaker already has defended the edge $e$. A pairing is then said to be {\em complete} in a position $\pos=(\hyp,M,B)$ if it is complete in the hypergraph with vertex set $V(\pos)$ and edge set $E(\pos)$. We have the following result.

\begin{proposition}[Pairing strategy]\label{prop:pairing}
    Let $a,b \in \Zp$. If a pairing $\Pi$ is complete in a position $\pos=(\hyp,M,B)$ of the $(a,b)$-game and $b-|B| \geq \min(a,|\Pi \cap 2^{V(\pos)}|)$, then Breaker wins the $(a,b)$-game played from that position (with Maker playing first, as usual).
\end{proposition}

\begin{proof}
    Breaker leaves his tokens on $B$ throughout, and only uses his $b-|B|$ other tokens. Whenever Maker places a token on some $x \in \{x,y\} \in \Pi$ where $y$ is free, Breaker places a token on $y$. This is obviously possible if $b-|B| \geq |\Pi \cap 2^{V(\pos)}|$, without even the need to move tokens. If $a \leq b-|B| < |\Pi \cap 2^{V(\pos)}|$, then Breaker may need to move tokens, but he has enough of them to ``follow'' Maker without ever vacating a pair on which Maker has one of her own. Using this strategy, Breaker ensures that Maker never has a token on any vertex in $B$ and never has tokens on both vertices of a pair in $\Pi \cap 2^{V(\pos)}$. By definition of a complete pairing, this implies that Maker never fills an edge of $\hyp$.
\end{proof}

As a first application of pairing strategies, we can easily solve token positional games on 2-uniform hypergraphs. Note that hypergraphs containing an edge of size 1 may be disregarded in general, since Maker wins in one move in that case.

\begin{proposition}\label{prop:2-unif}
    Let $\hyp$ be a 2-uniform hypergraph, and let $a,b \in \Zp$ with $a \geq 2$. Breaker wins the $(a,b)$-game on $\hyp$ if and only if the edges of $\hyp$ are pairwise disjoint and $b \geq \min(a,|E(\hyp)|)$.
\end{proposition}

\begin{proof}
    In the case where $\hyp$ has two edges of the form $\{x,y\}$ and $\{y,z\}$, Maker can start by placing a token on $y$ and win the game with her next move by placing one on either $x$ or $z$. Now, suppose that $\hyp$ has pairwise disjoint edges $\{x_i,y_i\}$. If $b \geq \min(a,|E(\hyp)|)$, then Proposition \ref{prop:pairing} ensures that Breaker wins using the pairing $E(\hyp)$ which is obviously complete in $\hyp$. If $b < \min(a,|E(\hyp)|)$, then Maker can start by placing distinct tokens on $x_1,\ldots,x_b$, which forces Breaker to place distinct tokens on $y_1,\ldots,y_b$ respectively as an answer, then Maker can place a new token on $x_{b+1}$: necessarily, some $y_i$ with $i \in \interval{1}{b+1}$ will be unoccupied after Breaker's turn, so Maker can move one of her tokens to $y_i$ and win the game.
\end{proof}

Finally, let us mention a construction which can help generalizing results to all edge sizes.

\begin{proposition}\label{prop:thetaplus1}
    For any $k$-uniform hypergraph $\hyp$ on $n$ vertices, there exists a $(k+1)$-uniform hypergraph $\hyp'$ on $n+2$ vertices satisfying the following two properties:
    \begin{itemize}[nolistsep,noitemsep]
        \item For all $a,b \in \Zp$, Maker wins the $(a,b)$-game on $\hyp$ if and only if Maker wins the $(a+1,b+1)$-game on $\hyp'$.
        \item For all $t \in \Zp$, Maker has a strategy ensuring that she fills an edge during the first $t$ rounds of the $(\ast,\ast)$-game on $\hyp$ if and only if Maker has a strategy ensuring that she fills an edge during the first $t+1$ rounds of the $(\ast,\ast)$-game on $\hyp'$.
    \end{itemize}
    
\end{proposition}

\begin{proof}
    Let $\hyp'$ be defined by: $V(\hyp')=V(\hyp) \cup \{v,\ov{v}\}$ where $v$ and $\ov{v}$ are new vertices, and $E(\hyp')=E_1 \cup E_2$ where $E_1 = \{e \cup \{v\} \mid e\in E(\hyp)\}$ and $E_2 = \{\{v,\ov{v}\} \cup U \mid U \subseteq V(\hyp), |U|=k-1\}$. It is clear that $\hyp'$ is $(k+1)$-uniform. Now, let $a,b \in \Zp$, and consider the $(a+1,b+1)$-game played on $\hyp'$. Recall that the $(a,b)$-game and the $(a+1,b+1)$-game both coincide with the $(\ast,\ast)$-game if $a$ and $b$ are large enough, so that we can address both assertions of the proposition at once. We claim that Maker and Breaker should place tokens on $v$ and $\ov{v}$ respectively during the first round and never move these two tokens afterwards, otherwise they get a losing position. Indeed:
    \begin{itemize}[nolistsep,noitemsep]
        \item[--] If $v$ is ever available to Breaker, then he wins the game by placing a token on $v$ and keeping it there afterwards since every edge of $\hyp'$ contains $v$.
        \item[--] After Maker has placed a token on $v$ as her first move, there are two possibilities. If $\min\left(a+1,\frac{|V(\hyp')|}{2}\right)\leq k$, then Maker will place at most $k$ tokens in total, which is not enough to fill an edge of size $k+1$, so Breaker could play arbitrary moves and win. If $\min\left(a+1,\frac{|V(\hyp')|}{2}\right)> k$ and $\ov{v}$ is ever available to Maker, then she wins by placing a token on $\ov{v}$ and then placing her remaining tokens arbitrarily inside $V(\hyp)$: indeed, since $|V(\hyp)|\geq 2k-1$ and Maker has $a-1 \geq k-1$ tokens to place outside $\{v,\ov{v}\}$, Maker will get tokens on $k-1$ vertices in $V(\hyp)$ and automatically fill an edge in $E_2 \subseteq E(\hyp')$.
    \end{itemize}
    Therefore, we can assume that Maker and Breaker place tokens on $v$ and $\ov{v}$ respectively during the first round, and that they never move these two tokens throughout. After that first round of play, Maker and Breaker have $a$ and $b$ remaining tokens to play with respectively. Moreover, all edges in $E_2$ contain $\ov{v}$ so Maker will never fill an edge in $E_2$, and all edges in $E_1$ contain $v$ so Maker filling an edge in $E_1$ is equivalent to Maker filling an edge in $E(\hyp)$. All in all, the situation after the first round of play boils down to the $(a,b)$-game on $\hyp$.
\end{proof}

\section{When Breaker has a single token}\label{section3}

We say a pair $(e_1,e_2)$ of distinct edges is {\em $a$-reducible} if $|e_1 \cap e_2| \geq |e_1|+|e_2|-a-2$. Note that we may have $e_1 \cap e_2 = \varnothing$.

\begin{lemma}\label{lem:reduced1}
     Let $\hyp$ be a hypergraph, and let $a \in \Zp$. If Maker wins the $(a,1)$-game on $\hyp$, then there exists either an edge of size at most 1 or an $a$-reducible pair of edges in $\hyp$.
\end{lemma}

\begin{proof}
    If no edge has size 0 or 1, then Maker needs tokens on all but one of the vertices of two distinct edges $e_1$ and $e_2$ at some point, to create a double threat. Since Maker has $a$ tokens, this can only happen if $a \geq (|e_1|-1)+(|e_2|-1)-|e_1 \cap e_2|$.
\end{proof}

\begin{lemma}\label{lem:reduced2}
    Let $\hyp$ be a hypergraph containing an $a$-reducible pair of edges $(e_1,e_2)$, and let $\hyp'$ be the ``reduced" hypergraph defined by $V(\hyp')=V(\hyp)$ and $E(\hyp')=(E(\hyp)\setminus\{e_1,e_2\})\cup\{e_1 \cap e_2\}$. Then $\hyp$ and $\hyp'$ have the same outcome for the $(a,1)$-game.
\end{lemma}

\begin{proof}
    If Maker wins the $(a,1)$-game on $\hyp$, then the same winning strategy also works on $\hyp'$ since each edge of $\hyp'$ is a subset of some edge of $\hyp$. Conversely, suppose Maker wins the $(a,1)$-game on $\hyp'$ and applies that winning strategy on $\hyp$. She will either win on $\hyp$, or get tokens on all vertices of $e_1 \cap e_2$. In the latter case, she keeps these tokens on $e_1 \cap e_2$ throughout and, in every further round, she uses one of her other tokens to place it on whichever of $e_1$ or $e_2$ does not currently contain Breaker's token. At some point, she will have tokens on all but one of the vertices of, say, $e_1$, and will not move them until the very last move of the game, thus forcing Breaker to place his token on the last vertex of $e_1$ and keep it there until the end. Maker then places tokens on all but one of the vertices of $e_2$, and finally uses a token that is not on $e_2$ to move it on the final vertex of $e_2$ and win the game.
\end{proof}

\begin{theorem}\label{theo:breaker1}
    Deciding the outcome of the $(a,1)$-game can be done in polynomial time $O(m^3)$ where $m$ is the number of edges, even if $a$ is part of the input.
\end{theorem}

\begin{proof}
    While there exists an $a$-reducible pair of edges, perform the reduction from Lemma \ref{lem:reduced2}, which is outcome-neutral. At the end of this process, Lemma \ref{lem:reduced1} concludes: if there exists an edge of size 0 or 1, then Maker wins, otherwise Breaker wins. Finding an $a$-reducible pair is done in $O(m^2)$ time, and at most $m-1$ reductions are performed in total.
\end{proof}

Having solved token positional games where Breaker has a single token, we now go further by determining the minimum number of edges that Maker needs to win when her number of tokens equals the size of the edges.

\begin{proposition}
    Let $\hyp$ be a $k$-uniform hypergraph with $k \in \Zp$. If Maker wins the $(k,1)$-game on $\hyp$, then $|E(\hyp)|\geq \left\lfloor \frac{k}{2} \right\rfloor + 1$.
\end{proposition}

\begin{proof}
    Assume $k \geq 2$, as the result is trivial for $k=1$. It suffices to show that at least $\left\lfloor \frac{k}{2} \right\rfloor$ reductions from Lemma \ref{lem:reduced2} need to be performed on $\hyp$ before an edge of size 0 or 1 is created. Let $r(p)$ denote the minimum number of reductions that must be made from $\hyp$ to get an edge of size at most $p$. We want to show that $r(1) \geq \left\lfloor \frac{k}{2} \right\rfloor$ i.e. $r(1) \geq \frac{k-1}{2}$. We actually claim that:
    $$r(p) \geq \frac{k-p}{2}, \text{\quad for all $p \in \interval{1}{k}$}.$$
    We prove this claim by induction on $k-p$. For $p=k$, we obviously have $r(k)=0$ because $\hyp$ is $k$-uniform. Now, let $p\in\interval{1}{k-1}$, and assume that $r(p') \geq \frac{k-p'}{2}$ for all $p' \in \interval{p+1}{k}$. Suppose an edge $e=e_1 \cap e_2$ of size at most $p$ is created from a $k$-reducible pair $(e_1,e_2)$. Defining $p_1=|e_1|$ and $p_2=|e_2|$, we have $p \geq |e| = |e_1 \cap e_2| \geq |e_1|+|e_2|-k-2 = p_1+p_2-k-2$. Since the creations of $e_1$ and $e_2$ have needed at least $r(p_1)$ and $r(p_2)$ reductions themselves respectively, and using the induction hypothesis, we get:
    $$ r(p) \geq 1 + r(p_1)+r(p_2) \geq 1 +\frac{k-p_1}{2} + \frac{k-p_2}{2} \geq 1 +\frac{k-p_1}{2} + \frac{p_1-p-2}{2} = \frac{k-p}{2}. $$
    This proves the claim and the proposition.
\end{proof}

\begin{proposition}\label{prop:k-unif-kvs1}
    For all $k \geq 1$, there exists a $k$-uniform hypergraph $\hyp$ with $|E(\hyp)| = \left\lfloor \frac{k}{2} \right\rfloor + 1$ such that Maker wins the $(k,1)$-game on $\hyp$.
\end{proposition}

\begin{proof}
    We define $\hyp$ as follows (see Figure \ref{fig:k-unif-kvs1}):
    $$V(\hyp)=\{u_1,\ldots,u_k\} \cup \bigcup_{i=1}^{\left\lfloor \frac{k}{2} \right\rfloor} \left\{ v^{(i)}_1,\ldots,v^{(i)}_{k-1-2\left(\left\lfloor \frac{k}{2} \right\rfloor - i \right)} \right\}
    \text{\quad and \quad}
    E(\hyp)=\{e_0,e_1,\ldots,e_{\left\lfloor \frac{k}{2} \right\rfloor}\},$$
    where $e_0=\{u_1,\ldots,u_k\}$ and
    $e_i= \left\{ v^{(i)}_1,\ldots,v^{(i)}_{k-1-2\left(\left\lfloor \frac{k}{2} \right\rfloor - i \right)}, u_{k-2\left(\left\lfloor \frac{k}{2} \right\rfloor - i \right)},\ldots,u_k \right\}$
    for all $i\in\interval{1}{\left\lfloor \frac{k}{2} \right\rfloor}$.
    \\ Note that $\hyp$ is $k$-uniform. We now perform reductions until we get an edge of size 1. Define $e'_i=\left\{ u_{k-2\left(\left\lfloor \frac{k}{2} \right\rfloor - i \right)} , \ldots , u_k \right\}$ for all $i\in\interval{1}{\left\lfloor \frac{k}{2} \right\rfloor}$. We have:
    \begin{itemize}
        \item $e_0 \cap e_1 = e'_1$, and $$|e_0 \cap e_1|=|e'_1| = 2 \left( \left\lfloor \frac{k}{2} \right\rfloor -1 \right) +1 \geq k-2 = |e_0|+|e_1|-k-2.$$
        \item For all $i\in\interval{2}{\left\lfloor \frac{k}{2} \right\rfloor}$: $e'_{i-1} \cap e_i = e'_i$, and $$|e'_{i-1} \cap e_i|=|e'_i| = 2 \left( \left\lfloor \frac{k}{2} \right\rfloor -i \right) +1 \geq 2 \left( \left\lfloor \frac{k}{2} \right\rfloor -(i-1) \right) +1 -2 = |e'_{i-1}|+|e_i|-k-2.$$
    \end{itemize}
    Therefore, $(e_0,e_1)$ is a $k$-reducible pair, and after replacing $e_0$ and $e_1$ with the edge $e'_1=e_0 \cap e_1$ we get $(e'_1,e_2)$ as a $k$-reducible pair, and after replacing $e'_1$ and $e_2$ with the edge $e'_2=e'_1 \cap e_2$ we get $(e'_2,e_3)$ as a $k$-reducible pair, etc. until we get the edge $e'_{\left\lfloor \frac{k}{2} \right\rfloor}=\{u_k\}$ of size 1. By Lemma \ref{lem:reduced2}, the outcome is preserved throughout the reductions, so Maker wins the $(k,1)$-game on $\hyp$.
\end{proof}

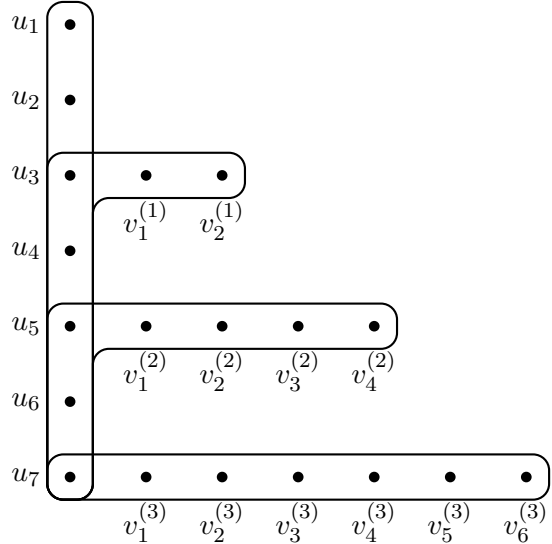
\begin{figure}[h]
    \centering
\begin{tikzpicture}[
    vertex/.style={circle, draw, fill=white, inner sep=1.5pt},
    hyperedge/.style={draw, thick, rounded corners=6pt}
]

\draw (0,0) node[v] (u7) {} node[left = .25]{$u_7$};
\draw (0,1) node[v] (u6) {} node[left = .25]{$u_6$};
\draw (0,2) node[v] (u5) {} node[left = .25]{$u_5$};
\draw (0,3) node[v] (u4) {} node[left = .25]{$u_4$};
\draw (0,4) node[v] (u3) {} node[left = .25]{$u_3$};
\draw (0,5) node[v] (u2) {} node[left = .25]{$u_2$};
\draw (0,6) node[v] (u1) {} node[left = .25]{$u_1$};

\draw (1,4) node[v] (v11) {} node[below = .2]{$v^{(1)}_1$};
\draw (2,4) node[v] (v12) {} node[below = .2]{$v^{(1)}_2$};

\draw (1,2) node[v] (v21) {} node[below = .2]{$v^{(2)}_1$};
\draw (2,2) node[v] (v22) {} node[below = .2]{$v^{(2)}_2$};
\draw (3,2) node[v] (v23) {} node[below = .2]{$v^{(2)}_3$};
\draw (4,2) node[v] (v24) {} node[below = .2]{$v^{(2)}_4$};

\draw (1,0) node[v] (v31) {} node[below = .2]{$v^{(3)}_1$};
\draw (2,0) node[v] (v32) {} node[below = .2]{$v^{(3)}_2$};
\draw (3,0) node[v] (v33) {} node[below = .2]{$v^{(3)}_3$};
\draw (4,0) node[v] (v34) {} node[below = .2]{$v^{(3)}_4$};
\draw (5,0) node[v] (v35) {} node[below = .2]{$v^{(3)}_5$};
\draw (6,0) node[v] (v34) {} node[below = .2]{$v^{(3)}_6$};

\draw[hyperedge]  (-0.3,-0.3) --  (-0.3,4.3) --  (2.3,4.3) --  (2.3,3.7) --  (0.3,3.7) --  (0.3,-0.3) --  cycle;
\draw[hyperedge]  (-0.3,-0.3) --  (-0.3,2.3) --  (4.3,2.3) --  (4.3,1.7) --  (0.3,1.7) --  (0.3,-0.3) --  cycle;
\draw[hyperedge]  (-0.3,-0.3) --  (-0.3,0.3) --  (6.3,0.3) --  (6.3,-0.3) --   cycle;
\draw[hyperedge]  (-0.3,-0.3) --  (-0.3,6.3) --  (0.3,6.3) --  (0.3,-0.3) --   cycle;

\end{tikzpicture}
\caption{The construction from Proposition~\ref{prop:k-unif-kvs1} for $k = 7$.}
    \label{fig:k-unif-kvs1}
\end{figure}

\section{When Breaker has unlimited tokens}\label{section4}

This section addresses the $(a,\ast)$-game. As such, we introduce the following notation: for any hypergraph $\hyp$, we define $\theta(\hyp)$ as the minimum $a$ such that Maker wins the $(a,\ast)$-game on $\hyp$. If Breaker wins the $(\ast,\ast)$-game on $\hyp$, then $\theta(\hyp)=\infty$.

We relate our study of $\theta$ with that of another hypergraph parameter, which has been studied a lot in the literature of positional games~\cite{chvatal1978,HEFETZ2007213,HEFETZ200939,Bonnet2016,Bonnet2017,BUJTAS202410}: for any hypergraph $\hyp$, let $\tau(\hyp)$ denote the minimum $t$ such that Maker has a strategy to win the $(\ast,\ast)$-game on $\hyp$ in at most $t$ rounds of play (that is, Maker fills an edge after having played at most $t$ moves). If Breaker wins the $(\ast,\ast)$-game on $\hyp$, then $\tau(\hyp)=\infty$.

When Maker wins the $(\ast,\ast)$-game on a hypergraph $\hyp$, we can think of two ways to change the rules and complicate her task. The first is to give her limited time to fill an edge: this corresponds to the study of $\tau(\hyp)$. The second is to limit her number of tokens (while Breaker's remains unlimited), and work out the least amount so that she still wins: this corresponds to the study of $\theta(\hyp)$. The two quantities are linked through the following inequalities, where $\textup{ark}(\hyp)$ denotes the {\em antirank} of $\hyp$ i.e. the size of its smallest edge.

\begin{proposition}\label{prop:theta}
    Let $\hyp$ be a hypergraph, and suppose that $\theta(\hyp)$ and $\tau(\hyp)$ are finite i.e. Maker wins the $(\ast,\ast)$-game on $\hyp$. Then $\textup{ark}(\hyp) \leq \theta(\hyp) \leq \tau(\hyp) \leq \left\lceil\frac{|V(\hyp)|}{2}\right\rceil$.
\end{proposition}

\begin{proof}
    Maker needs at least $\textup{ark}(\hyp)$ tokens to fill an edge, hence the first inequality. Moreover, Maker's strategy to win the $(\ast,\ast)$-game on $\hyp$ in at most $t$ rounds is also winning for the $(t,\ast)$-game on $\hyp$, using a new token for each move. Finally, it is obvious that all vertices have been claimed after Maker has made $\lceil|V(\hyp)|/2\rceil$ moves.
\end{proof}

It is easy to see that there exist hypergraphs for which all quantities from Proposition \ref{prop:theta} are equal: for instance, that is the case if $\hyp$ has $2k-1$ vertices and $\binom{2k-1}{k}$ edges corresponding to all possible subsets of size $k$. But how large can the gaps between these quantities be in general? We first note that the 2-uniform case is straightforward.

\begin{proposition}
    Let $\hyp$ be a 2-uniform hypergraph, and suppose that $\theta(\hyp)$ and $\tau(\hyp)$ are finite i.e. Maker wins the $(\ast,\ast)$-game on $\hyp$. Then $\theta(\hyp) = \tau(\hyp) = 2$.
\end{proposition}

\begin{proof}
    Since Maker wins the $(\ast,\ast)$-game on $\hyp$, there exist two intersecting edges by Proposition \ref{prop:2-unif}, so Maker wins in two rounds.
\end{proof}

The remainder of this section is dedicated to $k$-uniform hypergraphs for $k \geq 3$. In all figures, we will represent edges of size 3 or 4 as in Figure \ref{fig:edges}, with circled vertices corresponding to Maker's tokens and crossed out vertices corresponding to Breaker's tokens (to emphasize that these vertices are permanently out of the game from Maker's perspective).

\begin{figure}[h]
	\centering
	\includegraphics[scale=.58]{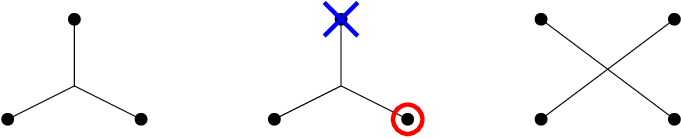}
	\caption{Left: an edge of size 3. Middle: an edge of size 3 on which Maker and Breaker each have a token. Right: an edge of size 4.}\label{fig:edges}
\end{figure}

\subsection{In 3-uniform hypergraphs}

A structural characterization of 3-uniform hypergraphs on which Maker wins the $(\ast,\ast)$-game is provided in \cite{GGS25}. It is based on two basic structures, represented in Figure \ref{fig:nunchaku_necklace}. A {\em nunchaku} of length $L \geq 1$ has vertex set $\{a_0,\ldots,a_L,b_1,\ldots,b_L\}$ and edge set $\{\{a_{i-1},b_i,a_i\}\mid i \in \interval{1}{L}\}$, such that the only tokens sitting on it are Maker tokens on $a_0$ and $a_L$. A {\em necklace} of length $L \geq 2$ has vertex set $\{a_1,\ldots,a_L,b_1,\ldots,b_L\}$ and edge set $\{\{a_i,b_i,a_{i+1}\}\mid i \in \interval{1}{L-1}\} \cup \{\{a_L,b_L,a_1\}\}$, such that the only token sitting on it is a Maker token on $a_1$. In the following theorem, the ``rounds of play'' need to be full rounds, meaning that we always look at the situation before Maker's turn.

\begin{figure}[h]
	\centering
	\includegraphics[scale=.58]{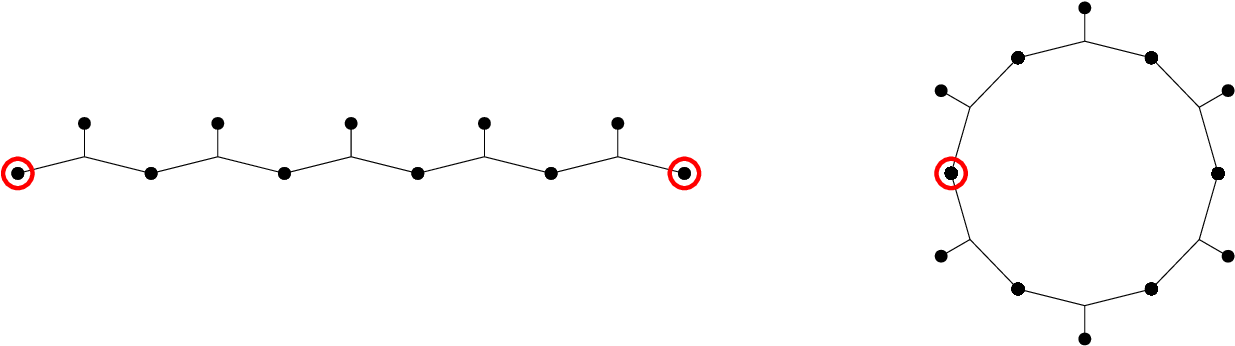}
	\caption{Left: a nunchaku of length 5. Right: a necklace of length 6.}\label{fig:nunchaku_necklace}
\end{figure}

\begin{theorem}[{\cite[Theorem 3.23]{GGS25}}]\label{theo:MB3}
    Let $\hyp$ be a 3-uniform hypergraph. Maker wins the $(\ast,\ast)$-game on $\hyp$ if and only if she has a strategy ensuring that the position obtained after at most three rounds of play contains a nunchaku or a necklace.
\end{theorem}

A corollary is that, when playing on a 3-uniform hypergraph against an unlimited amount of tokens for Breaker, Maker's number of tokens makes no difference: either she wins with just three tokens, or she does not win at all. 

\begin{corollary}\label{coro:theta3}
    Every 3-uniform hypergraph $\hyp$ satisfies $\theta(\hyp) \in \{3,\infty\}$.
\end{corollary}

\begin{proof}
    Assume that $\theta(\hyp)<\infty$ i.e. Maker wins the $(\ast,\ast)$-game on $\hyp$. By Theorem \ref{theo:MB3}, Maker has a strategy ensuring that the position obtained after at most three rounds of play contains a nunchaku or a necklace. Therefore, Maker can also reach such a position playing the $(3,\ast)$-game. We now show that Maker has a ``forcing strategy'' to force every move of Breaker until he is trapped. First suppose that there is a nunchaku with vertex set $\{a_0,\ldots,a_L,b_1,\ldots,b_L\}$, edge set $\{\{a_{i-1},b_i,a_i\}\mid i \in \interval{1}{L}\}$, and Maker tokens on $a_0$ and $a_L$. Maker moves her third token to $a_1$, threatening to fill the edge $\{a_0,b_1,a_1\}$ and thus forcing Breaker to claim $b_1$. Then, for all $i \in \interval{2}{L-1}$, Maker moves a token from $a_{i-2}$ to $a_i$, threatening to fill the edge $\{a_{i-1},b_i,a_i\}$ and thus forcing Breaker to claim $b_i$. However, since Maker's token on $a_L$ has not moved, she can now fill the edge $\{a_{L-1},b_L,a_L\}$ by moving a token from $a_{L-2}$ to $b_L$. This concludes the case where there is a nunchaku. If there is a necklace instead, the same forcing strategy works: Maker leaves the token of the necklace where it is, then uses her other two tokens to create threats along the cycle until we return to the unmoved token at which point there are two threats at once.
\end{proof}

The difference with the 2-uniform case is that the finite values of $\tau$ are not bounded. As such, there can be a gap between $\theta$ and $\tau$, and we know exactly how big it can be.

\begin{proposition}
    For all $n \geq 7$, there exists a 3-uniform hypergraph $\hyp$ on $n$ vertices such that $\theta(\hyp)=3$ and $\tau(\hyp)=\log_2(n)+O(1)$. Moreover, the finite values of $\tau$ are bounded by $\lceil\log_2(n)\rceil+3$ for general 3-uniform hypergraphs, so this maximizes the gap between $\theta$ and $\tau$ for 3-uniform hypergraphs up to an additive constant.
\end{proposition}

\begin{proof}
    Even though it is not a valid example strictly speaking, as it features tokens which are already placed, let us start by considering a nunchaku of length $L$ with vertex set $\{a_0,\ldots,a_L,b_1,\ldots,b_L\}$, edge set $\{\{a_{i-1},b_i,a_i\}\mid i \in \interval{1}{L}\}$, and Maker tokens on $a_0$ and $a_L$. The forcing strategy that we have seen wins in $\Theta(L)$ rounds, since it goes through the entire nunchaku from one end to the other. However, we claim that Maker can win much faster, namely in $1+\lceil \log_2(L) \rceil$ rounds, and that this is optimal. This result is proved in \cite[Lemma 5.6]{GGS25}, but we can easily explain it here since it uses a very simple dichotomy argument. It is easy to see that claiming some $b_i$ is a losing move for Maker, as Breaker could claim $a_{i-1}$ or $a_i$ next and destroy Maker's attack. Therefore, Maker should claim some $a_i$, thus effectively dividing the nunchaku into two smaller nunchakus. Breaker cannot play inside both of these nunchakus at once, so his next move will leave one of them intact for Maker to play in. If the two nunchakus are of different size, then Breaker should play inside the smaller one. As such, Maker's first move should be to claim $a_{\lceil L/2 \rceil}$, so as to divide the nunchaku into two nunchakus whose lengths differ by at most 1. By repeatedly claiming the $a_i$ which is in the middle of the nunchaku that Breaker has not played in, Maker eventually wins in a total of exactly $1+\lceil \log_2(L) \rceil$ rounds.
    
    To get a starting hypergraph which is token-free, one may simply take a nunchaku and replace each token with a ``diamond'' as in Figure \ref{fig:nunchaku_diamond}. Note that, if Maker claims a vertex of degree 3 and Breaker does not immediately answer inside the corresponding diamond, then Maker can win in two more moves. This construction on $n$ vertices is possible for any odd $n \geq 7$ (the case $n=7$ is the one where the two diamonds share a vertex), and can be extended to any even $n \geq 8$ by adding an isolated vertex. The dichotomy argument is still valid on this hypergraph $\hyp$, so that $\tau(\hyp)=\log_2(n)+O(1)$. On the other hand, Corollary \ref{coro:theta3} yields $\theta(\hyp) = 3$.

    \begin{figure}[h]
		\centering
		\includegraphics[scale=.58]{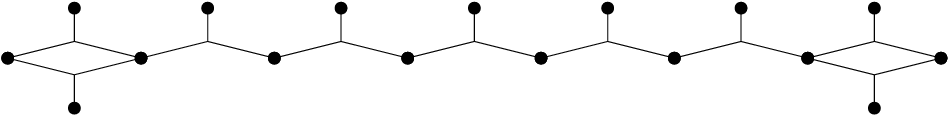}
		\caption{A token-free ``equivalent'' of the nunchaku.}\label{fig:nunchaku_diamond}
	\end{figure}
    
    As for the final statement of this proposition, it is a consequence of Theorem \ref{theo:MB3}. Indeed, if Maker wins the $(\ast,\ast$)-game on a 3-uniform hypergraph, then she has a strategy ensuring that the position obtained after at most three rounds of play contains a nunchaku or a necklace. Since she can win in $1+\lceil \log_2(L) \rceil$ rounds on a nunchaku of length $L$ (or a necklace of length $L$, with the same proof), she can win in a total number of rounds of at most $3+(1+\lceil \log_2(L) \rceil) = \lceil \log_2(L) \rceil + 4$. Finally, since a nunchaku (resp. a necklace) of length $L$ has $2L+1$ (resp. $2L$) vertices, we have $L \leq \frac{n}{2}$ hence $\lceil \log_2(L) \rceil + 4 \leq \lceil\log_2(n)\rceil+3$.
\end{proof}

\subsection{In $k$-uniform hypergraphs with $k \geq 4$}

We have just seen that, for 3-uniform hypergraphs, the gap between $\theta$ and $\tau$ can be arbitrarily large but no more than logarithmic. We now show that this gap can be much bigger for $k$-uniform hypergraphs with $k \geq 4$, as there are then examples where the lower bound for $\theta$ and the upper bound for $\tau$ given by Proposition \ref{prop:theta} are both exactly attained. Our construction generalizes the necklace for edges of size more than 3, except that the dichotomy strategy does not work anymore and Maker has to employ the (very slow) forcing strategy.

\begin{proposition}\label{prop:gap}
    For all $k \geq 4$ and all $n \geq 2k+1$, there exists a $k$-uniform hypergraph $\hyp$ on $n$ vertices such that $\theta(\hyp)=k$ and $\tau(\hyp)=\left\lceil \frac{n}{2} \right\rceil$. Moreover, this maximizes the gap between $\theta$ and $\tau$.
\end{proposition}

\begin{proof}
    Let us first show that addressing the case $k=4$ is sufficient. Suppose that the result holds for $k=4$. Let $k \geq 5$, and let $n \geq 2k+1$. Since $n-2(k-4) \geq 9$, we know there exists a 4-uniform hypergraph $\hyp$ on $n-2(k-4)$ vertices such that $\theta(\hyp)=4$ and $\tau(\hyp)=\left\lceil \frac{n-2(k-4)}{2} \right\rceil$. Through $k-4$ consecutive applications of Proposition \ref{prop:thetaplus1}, we get a $k$-uniform hypergraph $\hyp'$ on $n$ vertices such that $\theta(\hyp')=\theta(\hyp)+(k-4) = k$ and $\tau(\hyp')=\tau(\hyp)+(k-4) = \left\lceil \frac{n-2(k-4)}{2} \right\rceil + (k-4) = \left\lceil \frac{n}{2} \right\rceil$.
    
    Now, for the case $k=4$, we may assume that $n$ is odd: indeed, if $n$ is even, then we simply add an isolated vertex to our construction for $n-1$. Define $L=\frac{n-3}{2}$ and the 4-uniform hypergraph $\hyp$ as follows (see Figure \ref{fig:biggap}):
     \begin{align*}
         V(\hyp) & = \{a_1,\ldots,a_{L+1},b_1,\ldots,b_L,\ov{a_1},\ov{a_2}\} \cup \{u_1,u_2,u_3,u_4,u_5\};\\
         E(\hyp) & = \{e_1,\ldots,e_L\} \cup E_1 \cup E_2, \,\,\,\text{where} \\
         e_i & = \{a_i,a_{i+1},a_{i+2},b_i\} \,\,\,\text{for all $i \in \interval{1}{L-1}$,}\,\,\,
         e_L = \{a_L,a_{L+1},a_1,b_L\}, \,\,\,\text{and}\\
         E_i & = \{ \{a_i,\ov{a_i},u_j,u_{j'}\} \mid j \neq j' \in \interval{1}{5}\} \,\,\,\text{for all $i \in \interval{1}{2}$}.
     \end{align*}
     
     \begin{figure}[h]
		\centering
		\includegraphics[scale=.55]{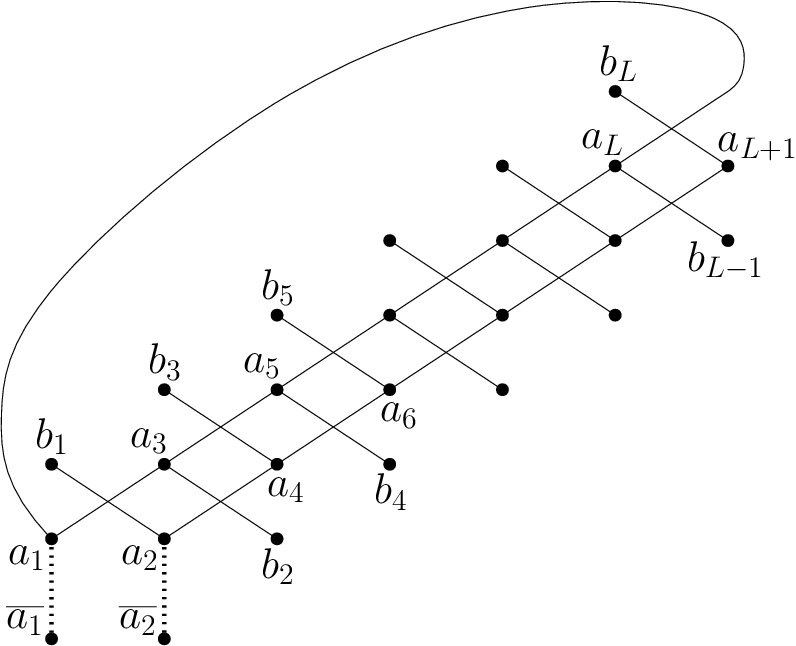}
		\caption{The construction for Proposition \ref{prop:gap}, with $L=11$. The two dotted lines at the bottom symbolize the edges in $E_1$ and $E_2$ respectively.}\label{fig:biggap}
	\end{figure}
	
     To minimize the number of vertices used in this construction, we actually choose $u_1,u_2,u_3,u_4,u_5$ arbitrarily among the $2L-1$ vertices $a_3,\ldots,a_{L+1},b_1,\ldots,b_L$ (as made possible by the fact that $n \geq 9$ hence $L \geq 3$), so that $\hyp$ indeed has $(L+1)+L+2=n$ vertices. However, we prefer to name them differently, as imagining that the $u_j$'s are separate from the other vertices should help the reader. The edges in $E_1 \cup E_2$, which are the only ones containing some of the $u_j$'s, will be irrelevant after the first phase of play anyway, as they are simply here to allow Maker to claim $a_1$ and $a_2$ uncontested at the start of the game. The set $E_i$ is there to emulate an edge $\{a_i,\ov{a_i}\}$, since we cannot use such an edge in a 4-uniform construction.
     
     For the $(\ast,\ast)$-game on $\hyp$, we now define {\em regular play}, a specific unfolding of the game which sees Maker win in a total of exactly $\frac{n+1}{2}=L+2$ moves and using only four tokens:
     \begin{itemize}[nolistsep,noitemsep]
         \item Phase 1: Maker places tokens on $a_1$ and $a_2$ in no particular order, which Breaker answers by claiming $\ov{a_1}$ and $\ov{a_2}$ respectively.
         \item Phase 2: Maker employs a forcing strategy along the edges $e_1,\ldots,e_L$ in that order, making one-move threats until the end of the game similarly to what we have seen in the 3-uniform case. She starts by placing a token on $a_3$, which prompts Breaker to claim $b_1$ because of the edge $e_1$, then she places a token on $a_4$, which prompts Breaker to claim $b_2$ because of the edge $e_2$. Maker now has four tokens on the board. The token on $a_1$ remains there, and for all $i \in \interval{3}{L-1}$, Maker moves the token on $a_{i-1}$ to put it on $a_{i+2}$, which prompts Breaker to claim $b_i$ because of the edge $e_i$. Finally, Maker moves the token on $a_{L-1}$ to put it on $b_L$ and wins by filling the edge $e_L$.
     \end{itemize}
     We now show that regular play is optimal for both players, where optimality should be understood as Maker trying to win as fast as possible and Breaker trying to survive for as long as possible. Despite Maker only using four tokens throughout regular play, we show that she could not win any faster using more tokens.
     
     \begin{claim}\label{cla:k-necklace1}
         If Maker deviates from regular play after both players had followed regular play from the beginning (even if Maker had used a new token for each move), then Maker gets a losing position.
     \end{claim}
     
     \begin{claim}\label{cla:k-necklace2}
         It is optimal for Breaker to follow regular play on his turn if both players have followed regular play so far.
     \end{claim}
     
     \begin{proofclaim}[Proof of Claim \ref{cla:k-necklace2} assuming Claim \ref{cla:k-necklace1}]
         We start with the easier proof, assuming Claim \ref{cla:k-necklace1} holds. On his turn, Breaker has the option to keep following regular play for as long as Maker also does: if Maker ever deviates, then Breaker wins by Claim \ref{cla:k-necklace1}, otherwise Breaker loses in $\frac{n+1}{2}$ moves which is the slowest possible loss. Therefore, the only way that deviating from regular play could be a better option for Breaker is if it allows him to win the game. As such, it suffices to show that deviating from regular play leads to Breaker losing the game. This is clear in Phase 2, as Breaker is simply defending Maker's successive threats of winning in one move. As for Phase 1, suppose that Maker has just claimed $a_i$ for some $i \in\interval{1}{2}$ (possibly after having already claimed $a_{3-i}$ and Breaker having responded by claiming $\ov{a_{3-i}}$ as per regular play), and suppose that Breaker answers by claiming some $x \neq \ov{a_i}$. Maker can now claim $\ov{a_i}$ herself. At this point, at least four of the $u_j$'s are still available (there are five of them in total, but $x$ might have been one). Therefore, Maker may claim any two arbitrary vertices among the $u_j$'s with her next two moves, thus automatically filling some edge in $E_i$.
     \end{proofclaim}
     
     \begin{proofclaim}[Proof of Claim \ref{cla:k-necklace1}]
         We show that, if Maker deviates from regular play, then Breaker has a move after which the resulting position admits a complete pairing, so Breaker has a winning strategy from this point on by Proposition \ref{prop:pairing}.
         \begin{itemize}
             
             \item Case 1: Maker deviates from regular play during Phase 1, i.e. Maker claims some $x \not\in \{a_1,a_2\}$ without having claimed both $a_1$ and $a_2$ beforehand. If $x=\ov{a_i}$ for some $i \in\interval{1}{2}$, then Breaker claims $y=a_i$, otherwise Breaker claims an arbitrary available vertex $y\in \{a_1,a_2\}$. In all cases, we now exhibit a pairing $\Pi(x,y)$ which is complete in the resulting position (see Figure \ref{fig:biggap_case1}).
             
             \begin{itemize}
                 \item Case 1.a: $x=a_{i_0}$ for some $i_0 \in \interval{3}{L+1}$ and $y=a_1$. Since the edges in $E_1 \cup \{e_1,e_L\}$ all contain $y$, we just need to cover those in $E_2 \cup \{e_2,\ldots,e_{L-1}\}$. For this, we define:
                 $$ \Pi(x,y) = \{\{a_2,\ov{a_2}\}\} \cup \{\{a_j,b_{j-1}\}\mid j \in \interval{3}{i_0-1}\} \cup \{\{a_j,b_{j-2}\}\mid j \in \interval{i_0+1}{L+1}\}. $$
                 Note that it is possible that $a_2$ and $\ov{a_2}$ have already been claimed in the first round, by Maker and Breaker respectively, but it is not an issue since the edges in $E_2$ do not need to be covered in that case. Our definition of a pairing allows for including vertices that have already been claimed, and it can still be complete as long as those vertices are not needed to cover the necessary edges.
                 \item Case 1.b: $x=b_{i_0}$ for some $i_0 \in \interval{1}{L}$ and $y=a_1$. We define:
                 $$ \Pi(x,y) = \{\{a_2,\ov{a_2}\}\} \cup \left\{\{a_{2j-1},a_{2j}\} \,\middle|\, j \in\interval{2}{\left\lfloor\frac{L+1}{2}\right\rfloor} \right\}. $$
                 The pair of highest index is $\{a_L,a_{L+1}\}$ if $L$ is odd or $\{a_{L-1},a_L\}$ if $L$ is even, so it covers $e_{L-1}$.
                 \item Case 1.c: $x=a_{i_0}$ for some $i_0 \in \interval{3}{L+1}$ and $y = a_2$. Since the edges in $E_2 \cup \{e_1,e_2\}$ all contain $y$, we just need to cover those in $E_1 \cup \{e_3,\ldots,e_L\}$. We define:
                 $$ \Pi(x,y) = \{\{a_1,\ov{a_1}\}\} \cup \{\{a_j,b_j\}\mid j \in \interval{3}{i_0-1}\} \cup \{\{a_j,b_{j-1}\}\mid j \in \interval{i_0+1}{L+1}\}.$$
                 \item Case 1.d: $x=b_{i_0}$ for some $i_0 \in \interval{1}{L}$ and $y = a_2$. We define:
                 $$ \Pi(x,y) = \{\{a_1,\ov{a_1}\}\} \cup \left\{\{a_{L-2j},a_{L+1-2j}\} \,\middle|\, j \in\interval{0}{\left\lceil\frac{L}{2}\right\rceil-2} \right\}. $$
                 The pair of lowest index is $\{a_3,a_4\}$ if $L$ is odd or $\{a_4,a_5\}$ if $L$ is even, so it covers $e_3$.
                 \item Case 1.e: $x=\ov{a_{i_0}}$ for some $i_0 \in \interval{1}{2}$. Since $y=a_{i_0}$ in this case, the edges in $E_i$ do not need to be covered, so we may reuse one of the pairings defined in the previous cases e.g. $\Pi(x,y)=\Pi(a_3,a_{i_0})$.
             \end{itemize}
             
             \begin{figure}[h]
		        \centering
		        \includegraphics[scale=.55]{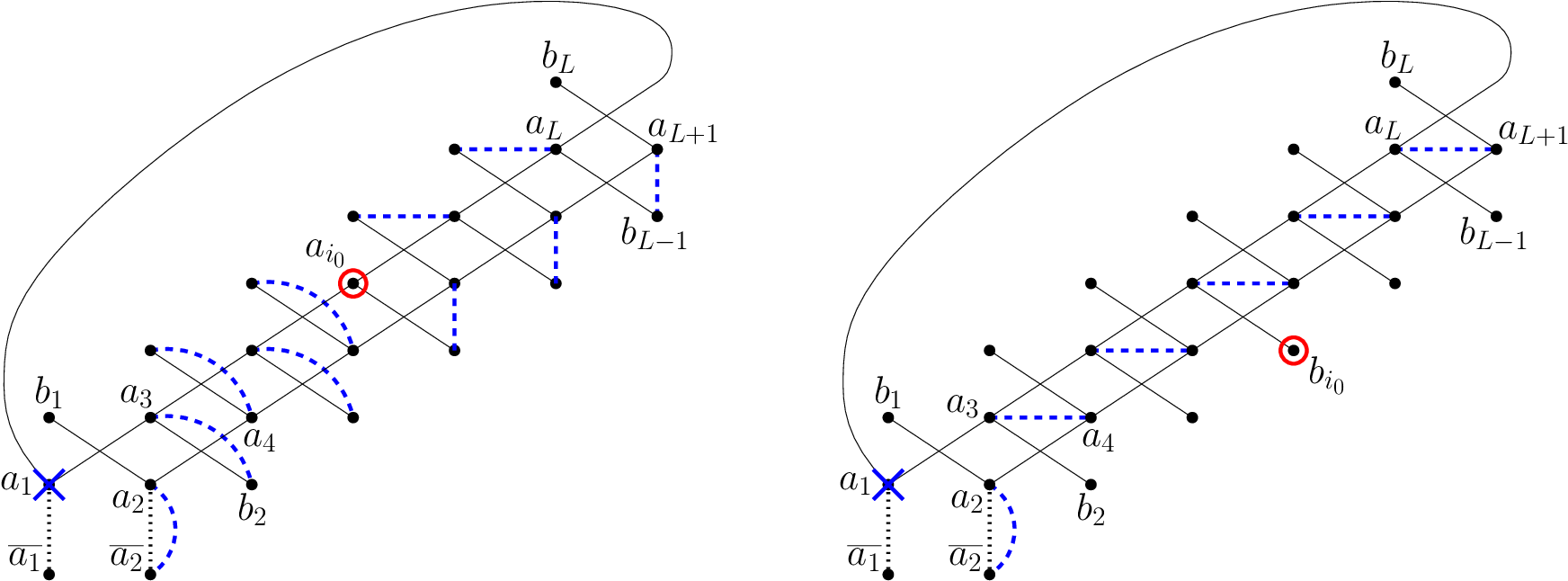}
		        \caption{Illustration of Cases 1.a (left) and 1.b (right) with $L=11$. The dashed blue lines represent the pairing.}\label{fig:biggap_case1}
	        \end{figure}
             
             \item Case 2: Maker deviates from regular play during Phase 2. Say this happens after exactly $i-1$ rounds of regular play have been played ($i \in \interval{3}{L+1}$) is unique), during which the successive vertices claimed by the players were $a_1,a_2,a_3,\ldots,a_{i-1}$ for Maker and $\ov{a_1},\ov{a_2},b_1,\ldots,b_{i-3}$ for Breaker. At this point, Maker deviates and claims some $x \neq a_i$ instead of claiming $a_i$ as per regular play. Now, Breaker claims $a_i$ himself. Note that the only edges that Breaker has not yet played in at this point are $e_{i+1},\ldots,e_L$. Again, in all cases, we exhibit a pairing $\Pi(x)$ which is complete in the resulting position (see Figure \ref{fig:biggap_case2}).
             
             \begin{itemize}
                 \item Case 2.a: $x=a_{i_0}$ for some $i_0 \in \interval{i+1}{L+1}$. We define:
                 $$\Pi(x) = \{\{a_j,b_j\}\mid j \in \interval{i+1}{i_0-1}\} \cup \{\{a_j,b_{j-1}\}\mid j \in \interval{i_0+1}{L+1}\}.$$
                 \item Case 2.b: $x=b_{i_0}$ for some $i_0 \in \interval{i-2}{L}$.
                 We define:
                 $$ \Pi(x) = \left\{\{a_{L-2j},a_{L+1-2j}\} \,\middle|\, j \in\interval{0}{\left\lceil\frac{L-i}{2}\right\rceil-1} \right\}.$$
                 The pair of lowest index is $\{a_{i+2},a_{i+3}\}$ if $L$ and $i$ have same parity or $\{a_{i+1},a_{i+2}\}$ otherwise, so it covers $e_{i+1}$. \qedhere
             \end{itemize}
             
             \begin{figure}[h]
		        \centering
		        \includegraphics[scale=.55]{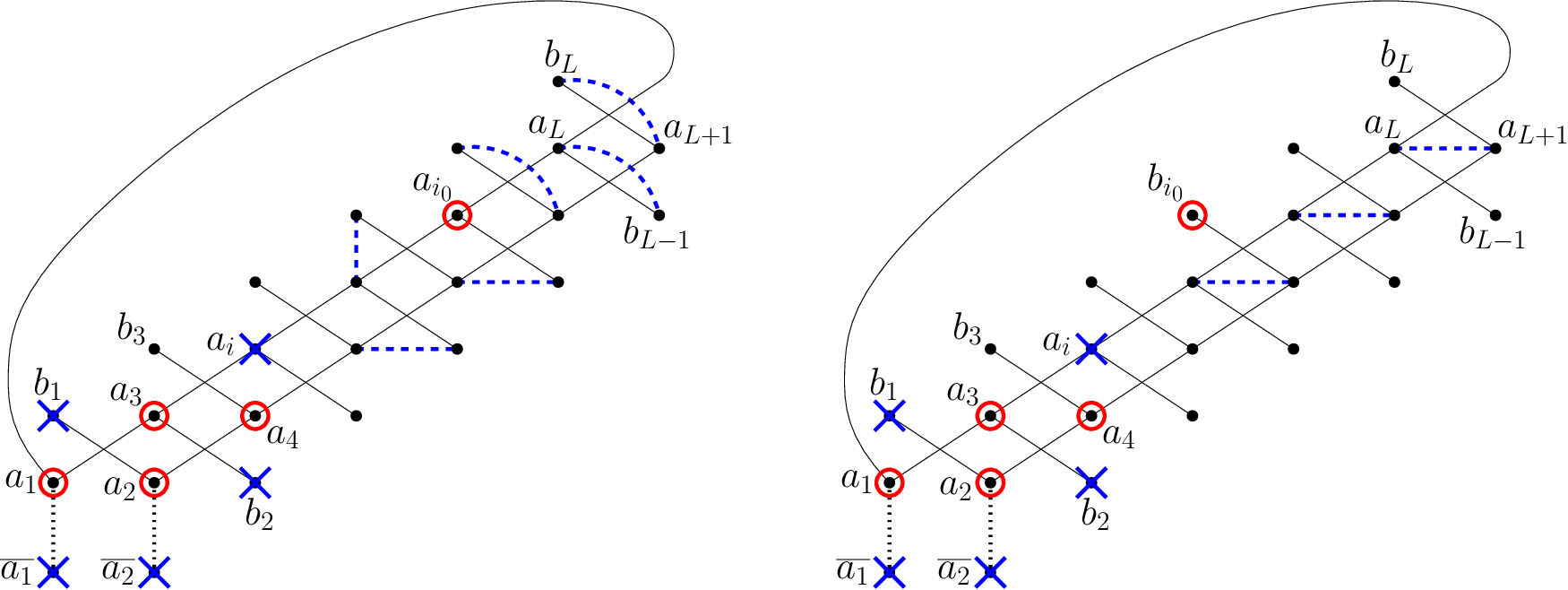}
		        \caption{Illustration of Cases 2.a (left) and 2.b (right) with $L=11$ and $i=5$. The dashed blue lines represent the pairing.}\label{fig:biggap_case2}
	        \end{figure}
	        
         \end{itemize}
     \end{proofclaim}
     Putting Claims \ref{cla:k-necklace1} and \ref{cla:k-necklace2} together, we get optimality of regular play from both players, hence $\theta(\hyp)=4$ and $\tau(\hyp)=\frac{n+1}{2}$. The final assertion of the proposition is then a direct consequence of Proposition \ref{prop:theta}.
\end{proof}

We have seen that, in 3-uniform hypergraphs, all but one of the vertices chosen by Maker only help her to make threats for one or two consecutive rounds and become useless afterwards, hence why she can manage with just three tokens. Things are very different in 4-uniform hypergraphs: there may be many vertices that Maker must hold on to for a long time, which can only be done if she has many tokens at her disposal. We actually show that the number of tokens that Maker needs to win can be linear in the total number of vertices.

\begin{proposition}\label{prop:bigtheta}
    For all $k \geq 4$ and all $n \geq 2k+12$, there exists a $k$-uniform hypergraph $\hyp$ on $n$ vertices such that $\theta(\hyp) \geq \left\lfloor\frac{n}{6}\right\rfloor$.
\end{proposition}

\begin{proof}
    Let us first show that addressing the case $k=4$ is sufficient. Suppose that the result holds for $k=4$. Let $k \geq 5$, and let $n \geq 2k+12$. Since $n-2(k-4) \geq 20$, we know there exists a 4-uniform hypergraph $\hyp$ on $n-2(k-4)$ vertices such that $\theta(\hyp) \geq \left\lfloor\frac{n-2(k-4)}{6}\right\rfloor$. Through $k-4$ consecutive applications of Proposition \ref{prop:thetaplus1}, we get a $k$-uniform hypergraph $\hyp'$ on $n$ vertices such that $\theta(\hyp')=\theta(\hyp)+(k-4) \geq \left\lfloor\frac{n-2(k-4)}{6}\right\rfloor + (k-4) = \left\lfloor\frac{n+4k-16}{6}\right\rfloor \geq \left\lfloor\frac{n}{6}\right\rfloor$.
    
    To prove the case $k=4$, we will show the following: for all $N \geq 2$, there exists a 4-uniform hypergraph $\hyp$ on $6N+8$ vertices such that $\theta(\hyp)=N+2$. Indeed, suppose that this statement is true. Given $n \geq 20$, we can then define $N=\left\lfloor\frac{n-8}{6}\right\rfloor$, and get a 4-uniform hypergraph $\hyp$ on $6N+8$ vertices such that $\theta(\hyp)=N+2$. Up to adding at most five isolated vertices to $\hyp$, we get a 4-uniform hypergraph $\hyp'$ on $n$ vertices such that $\theta(\hyp')=\theta(\hyp)=N+2 = \left\lfloor\frac{6N+13}{6}\right\rfloor \geq \left\lfloor\frac{n}{6}\right\rfloor$.
    
    Let $N \geq 2$. We define the 4-uniform hypergraph $\hyp$ as follows (see Figure \ref{fig:bigtheta}):
    \begin{align*}
        V(\hyp) & = \{a_0,\ldots,a_{2N},b_1,\ldots,b_{2N},c_1,\ldots,c_N,\ov{c_1},\ldots,\ov{c_N},\ov{a_0},\ov{a_{2N}},u_1,u_2,u_3,u_4,u_5\}; \\
        E(\hyp) & = \{e_1,\ldots,e_{2N}\} \cup E_{a_0} \cup E_{a_{2N}} \cup E_{c_1} \cup \cdots \cup E_{c_N}, \,\,\,\text{where} \\
         e_i & = \{a_{i-1},a_i,b_i,c_i\} \,\,\,\text{for all $i \in \interval{1}{N}$,}\,\,\, \\
         e_i & = \{a_{i-1},a_i,b_i,c_{i-N}\} \,\,\,\text{for all $i \in \interval{N+1}{2N}$,}\,\,\,\text{and}\\
         E_v & = \{ \{v,\ov{v},u_j,u_{j'}\} \mid j \neq j' \in \interval{1}{5}\} \,\,\,\text{for all $v \in \{a_0,a_{2N},c_1,\ldots,c_N\}$}.
    \end{align*}

    \begin{figure}[h]
		\centering
		\includegraphics[scale=.55]{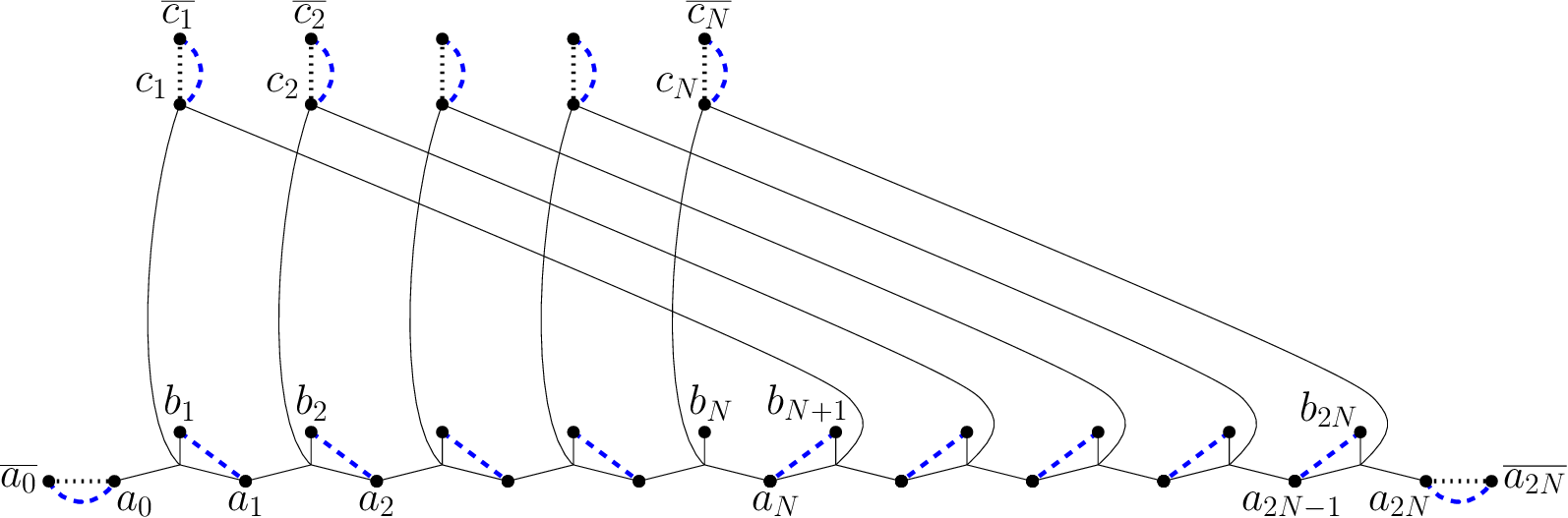}
		\caption{The construction for Proposition \ref{prop:bigtheta}. Dotted black lines symbolize the edges in the $E_v$'s. The dashed blue lines represent the pairing $\Pi$.}\label{fig:bigtheta}
	\end{figure}
            
    It is clear that $\hyp$ has $6N+8$ vertices. To conclude, we must show that $\theta(\hyp)=N+2$. The idea behind our construction is simple. Similarly to the $E_i$'s in the proof of Proposition \ref{prop:gap}, the $E_v$'s are there to allow Maker to place tokens on each $v\in\{a_0,a_{2N},c_1,\ldots,c_N\}$ uncontested: indeed, if Breaker does not answer by claiming the corresponding $\ov{v}$, then Maker can place a token on $\ov{v}$ herself and win in two more moves by placing tokens on any two of the vertices $u_1,u_2,u_3,u_4,u_5$ next. The upper bound on $\theta(\hyp)$ comes from the observation that, if Maker had tokens on $a_0,a_{2N},c_1,\ldots,c_N$, then she could leave them there and win using a forcing strategy and just one extra token, using the fact that the edges $e_1,\ldots,e_{2N}$ would then form a path behaving exactly like a nunchaku does in the 3-uniform case. As for the lower bound on $\theta(\hyp)$, the key is that the edges in each $E_v$ are ``single-use'': if Maker ever vacates the vertex $v$, then putting a token back on $v$ later will have no effect since Breaker had already claimed $\ov{v}$ the first time. As such, Maker needs to keep tokens on $c_1,\ldots,c_N$ throughout. We now proceed with the rigorous proofs.
    
    Let us start by showing that $\theta(\hyp) \leq N+2$, which is straightforward. Playing the $(N+2,\ast)$-game on $\hyp$, Maker starts by placing tokens on $a_0,c_1,\ldots,c_N$, forcing Breaker to claim $\ov{a_0},\ov{c_1},\ldots,\ov{c_N}$ respectively. Note that she does not actually put a token on $a_{2N}$ just yet, thus saving her last token, which she places on $a_1$. As Maker is threatening to fill $e_1$, Breaker is forced to claim $b_1$. After that, for all $i \in \interval{2}{2N}$, Maker moves her token on $a_{i-2}$ to put it on $a_i$, threatening to fill $e_i$ and forcing Breaker to claim $b_i$. Maker can then move any three tokens other than the one on $a_{2N}$, which she places successively on $\ov{a_{2N}}$ and on any two of $u_1,u_2,u_3,u_4,u_5$, thus filling some edge in $E_{a_{2N}}$ and winning the game.
    
    Let us finally show that $\theta(\hyp) \geq N+2$. We start by defining the following pairing (represented in Figure \ref{fig:bigtheta}), which is incomplete in $\hyp$ but covers every edge except for $e_N$:
    $$ \Pi = \{\{v,\ov{v}\} \mid v \in \{a_0,a_{2N},c_1,\ldots,c_N\}\} \cup \{\{a_i,b_i\} \mid i \in \interval{1}{N-1}\} \cup \{\{a_i,b_{i+1}\} \mid i \in \interval{N}{2N-1}\}.$$
    We will establish that the following two-phase strategy is a winning Breaker strategy for the $(N+1,\ast)$-game played on $\hyp$:
    \begin{itemize}[noitemsep,nolistsep]
        \item Phase 1: If Breaker can claim a vertex such that the resulting position admits a complete pairing, then he claims an arbitrary such vertex and switches to Phase 2. Otherwise, he applies the pairing strategy associated to the incomplete pairing $\Pi$. This means that, denoting by $x$ the vertex that Maker has just placed a token on, if $x$ is in a pair from $\Pi$ then Breaker answers by claiming its twin i.e. the other vertex of that pair, whereas if $x$ is in no pair from $\Pi$, or if it is but its twin has already been claimed by Breaker (which can happen since Maker can vacate a vertex and then re-place a token on it later), then Breaker claims an arbitrary vertex.
        \item Phase 2: Breaker applies the pairing strategy associated to a complete pairing until the end.
    \end{itemize}
    Let the $(N+1,\ast)$-game play out, with Maker employing any strategy of her choice and Breaker applying the above strategy. What remains to be seen is that Phase 2 is actually reached, since Proposition \ref{prop:pairing} then allows us to conclude.
    
    Whenever Maker moves a token that was already placed, we see it as two consecutive actions: Maker removes a token, then Maker places a token. It will help to consider the ``remove'' action as part of the previous round. Therefore, we define play during Round $t \geq 1$ as follows, in that order:
	\begin{enumerate}[noitemsep,nolistsep,label=\arabic*.]
		\item Maker places a token on some free vertex $x_t$. (We define $M_t \subseteq \{x_1,\ldots,x_t\}$ as the set of all vertices on which Maker has tokens at this point, with $M_0 = \varnothing$.)
		\item Breaker claims some free vertex $y_t$.
		\item If applicable, Maker removes one of her tokens. (We define $\pos_t$ as the position obtained at the end of Round $t$, with $\pos_0=(\hyp,\varnothing,\varnothing)$ being the starting position.) 
	\end{enumerate}
	
	Suppose for a contradiction that Maker wins. Let $T$ be the duration of the game i.e. Maker completes an edge during Round $T$. Note that $x_t,M_t$ are defined for $t \in \interval{0}{T}$ while $y_t,\pos_t$ are defined for $t \in \interval{0}{T-1}$. To account for moves being made from a given position $\pos=(\hyp,M,B)$, we introduce the notation $(\pos,M',B')=(\hyp,M \cup M', B \cup B')$.
	
	\begin{claim}\label{cla:bigtheta_pairing}
		All of Breaker's moves are made according to the pairing strategy associated to $\Pi$ and, for all $t \in \interval{1}{T}$ and for all free vertex $y$ in the position $(\pos_{t-1},\{x_t\},\varnothing)$, there is no pairing that is complete in the position $(\pos_{t-1},\{x_t\},\{y\})$. Moreover, for all $t \in\interval{0}{T-1}$, $\Pi$ covers all edges in $E(\pos_t)$ apart from $e_N$.
	\end{claim}
	\begin{proofclaim}[Proof of Claim \ref{cla:bigtheta_pairing}]
		Maker winning means Breaker is stuck in Phase 1 for the whole duration of the game. As for the last assertion of this claim, we know it holds for $t=0$, and the fact that Breaker applies the pairing strategy associated to $\Pi$ ensures that it remains true throughout the game.
	\end{proofclaim}
	
	The following claim consists in the simple observation that the edge filled by Maker to win is necessarily $e_N$, and that Maker only places a token on $b_N$ as her very last move.
	
	\begin{claim}\label{cla:bigtheta_last}
		We have $e_N \subseteq M_T$, moreover $x_1,\ldots,x_{T-1} \neq b_N$ and $x_T=b_N$.
	\end{claim}
	\begin{proofclaim}[Proof of Claim \ref{cla:bigtheta_last}]
		Since Breaker applies the pairing strategy associated to $\Pi$, which covers all edges in $E(\hyp)$ apart from $e_N$, the only edge that Maker can fill is $e_N$. Now, suppose for a contradiction that $x_t=b_N$ for some $t \in\interval{1}{T-1}$. Since $t<T$, there exists $y \in e_N \setminus M_t$: consider the position $\pos=(\pos_{t-1},\{b_N\},\{y\})$. We know by Claim \ref{cla:bigtheta_pairing} that $\Pi$ covers all edges in $E(\pos_{t-1})$ apart from $e_N$, moreover $e_N \not\in E(\pos)$ since $y \in e_N$, so $\Pi$ is complete in $\pos$. This contradicts Claim \ref{cla:bigtheta_pairing}.
	\end{proofclaim}
	
	Let us introduce the following notations, for all $i \in \interval{1}{N}$ (see Figure \ref{fig:intervals}):
	\begin{itemize}[noitemsep,nolistsep]
		\item $I_i^L = \{a_i,\ldots,a_{N-1}\} \cup \{b_i,\ldots,b_N\}$;
		\item $I_i^R = \{a_N,\ldots,a_{N+i-1}\} \cup \{b_N,\ldots,b_{N+i}\}$;
		\item $I_i = I_i^L \cup I_i^R$;
		\item $t(i) = \min\{t \in \interval{1}{T} \mid x_t \in I_i\}=\min\{t \in \interval{1}{T} \mid M_t \cap I_i \neq \varnothing\}$.
	\end{itemize}
	Note that $t(i)$ is well defined since $b_N \in I_i$ for all $i \in \interval{1}{N}$ and $b_N \in M_T$ by Claim \ref{cla:bigtheta_last}.
	\begin{figure}[h]
		\centering
		\includegraphics[scale=.58]{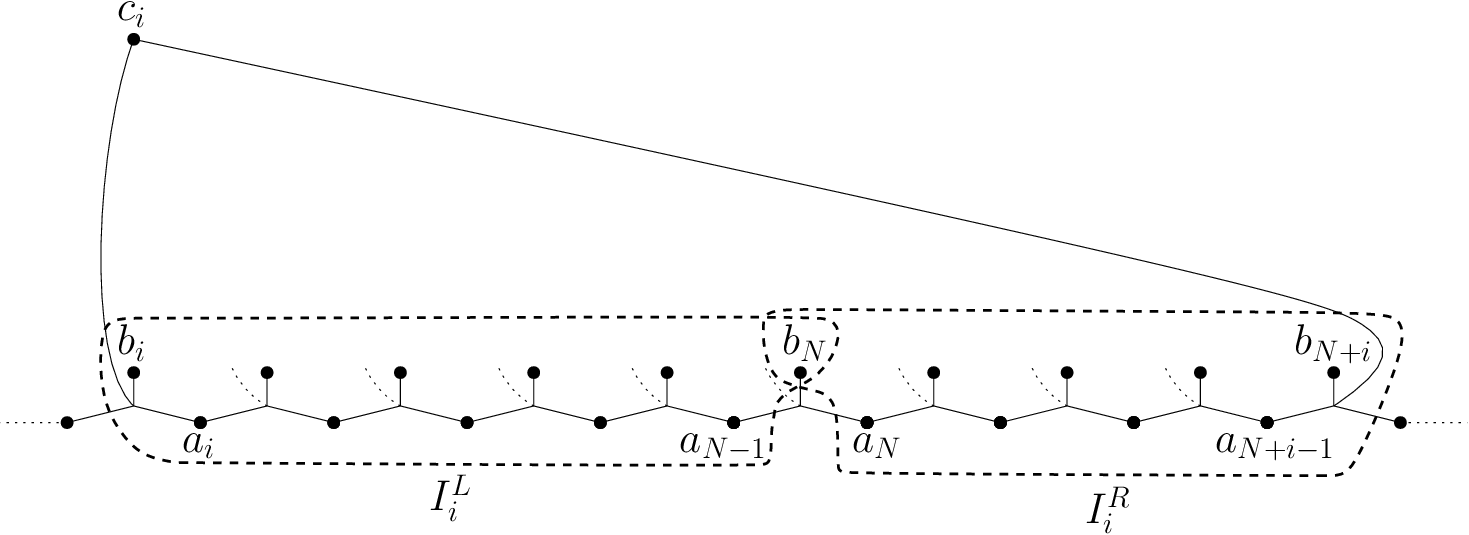}
		\caption{Definition of $I_i^L$ and $I_i^R$.}\label{fig:intervals}
	\end{figure}
	
	In the proof of the upper bound, we have seen how controlling the $c_i$'s (i.e. having tokens on $c_1,\ldots,c_N$) was key for Maker: it allowed her to force all of Breaker's moves and thus make progress from left to right until she won the game. The idea for the end of the proof relies on this principle, as follows. The first time that Maker places a token on some $x_t$ inside the ``interval" $I_i$, Maker must control $c_i$ (this will be Claim \ref{cla:bigtheta_control1}), otherwise Breaker could answer by claiming $c_i$ himself which breaks the path on both sides of $x_t$ and creates a complete pairing. Since $b_N \in I_i$ for all $i$, this means each $c_i$ is controlled by Maker at some point. However, Maker uses at most $N+1$ tokens in total, therefore she necessarily removes a token from some $c_i$ during the game (this will be Claim \ref{cla:bigtheta_control2}). The first time this happens, said $c_i$ is freed up to help building a complete pairing since Breaker has claimed $\ov{c_i}$ already. We now provide the details.
	
	\begin{claim}\label{cla:bigtheta_control1}
		For all $i \in\interval{1}{N}$, we have $c_i \in M_{t(i)-1}$.
	\end{claim}
	\begin{proofclaim}[Proof of Claim \ref{cla:bigtheta_control1}]
		Suppose for a contradiction that $c_i \not\in M_{t(i)-1}$ for some $i \in\interval{1}{N}$. Using the fact that there is no Maker token inside $I_i$ in $\pos_{t(i)-1}$ by minimality of $t(i)$, we are going to build a pairing $\Pi'$ that is complete in $\pos=(\pos_{t(i)-1},\{x_{t(i)}\},\{c_i\})$. This should be understood as: ``After Maker placed her token on $x_{t(i)}$, Breaker could have claimed $c_i$ and obtained a complete pairing''. We construct that pairing from $\Pi$ by modifying only the pairs inside $I_i$ (see Figure \ref{fig:claim_intervals}):
		\begin{itemize}[noitemsep,nolistsep]
			\item First case: $x_{t(i)}=b_i$ (the case $x_{t(i)}=b_{N+i}$ is analogous). We define $\Pi'$ to be the same as $\Pi$ except that the pairs inside $I_i$ are replaced by the pairs $\{a_{\ell},b_{\ell}\}$ for $\ell \in \interval{i+1}{N+i-1}$.
			\item Second case: $x_{t(i)}=b_j$ for some $j\in\interval{i+1}{N+i-1}$. We define $\Pi'$ to be the same as $\Pi$ except that the pairs inside $I_i$ are replaced by the pairs $\{a_{j-1},a_j\}$, $\{a_{\ell},b_{\ell+1}\}$ for $\ell \in \interval{i}{j-2}$, and $\{a_{\ell},b_{\ell}\}$ for $\ell \in \interval{j+1}{N+i-1}$.
			\item Third case: $x_{t(i)}=a_j$ for some $j\in\interval{i}{N+i-1}$. We define $\Pi'$ to be the same as $\Pi$ except that the pairs inside $I_i$ are replaced by the pairs $\{a_{\ell},b_{\ell+1}\}$ for $\ell \in \interval{i}{j-1}$ and $\{a_{\ell},b_{\ell}\}$ for $\ell \in \interval{j+1}{N+i-1}$.
		\end{itemize}
		In all cases, the newly defined pairs cover the edges $e_{i+1},\ldots,e_{N+i-1}$ (note that this includes $e_N$ which was the only edge not covered by $\Pi$). Moreover $e_i,e_{N+i} \not\in E(\pos)$ since $c_i \in e_i$ and $c_i \in e_{N+i}$, therefore $\Pi'$ covers all edges in $E(\pos)$. This contradicts Claim \ref{cla:bigtheta_pairing}.
	\end{proofclaim}
	
	\begin{figure}[h]
		\centering
		\includegraphics[scale=.58]{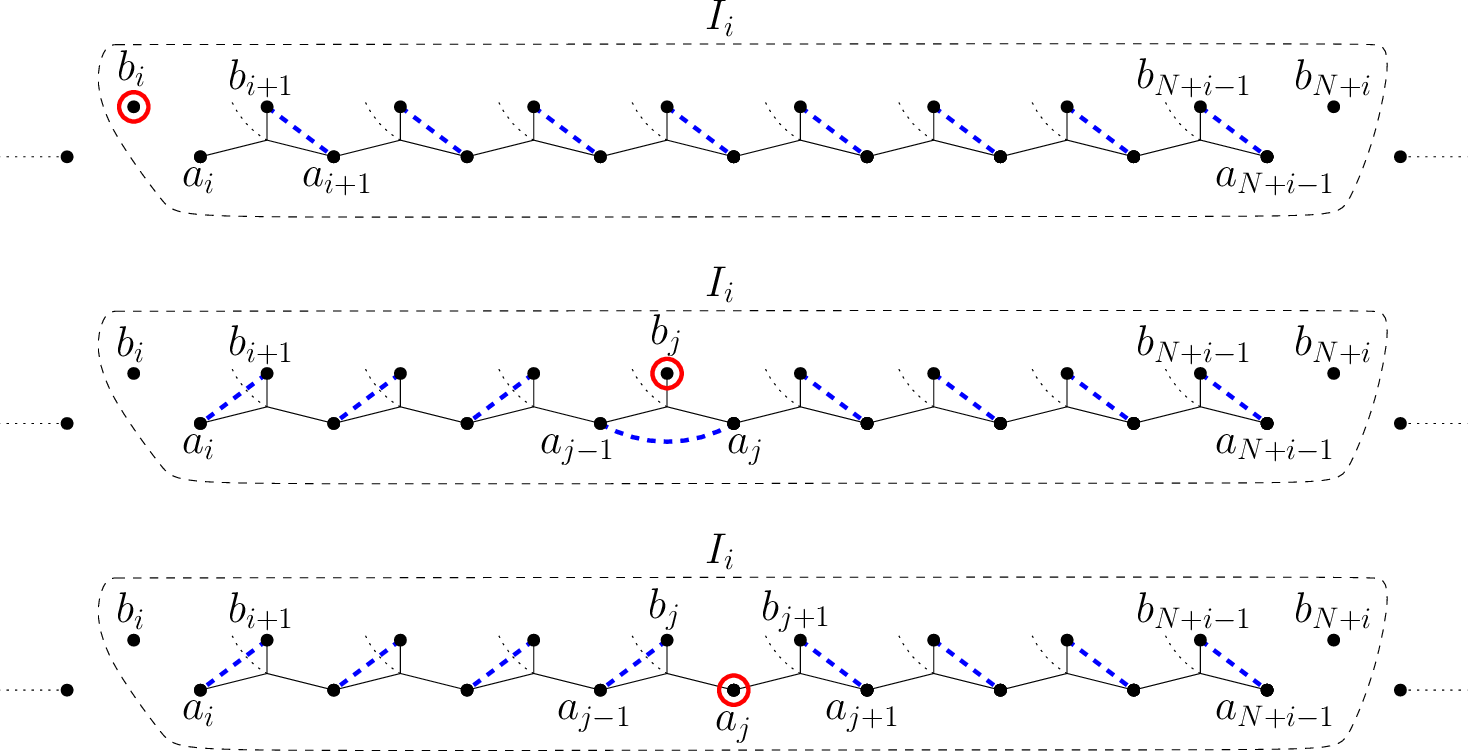}
		\caption{The pairing used inside $I_i$ in the proof of Claim \ref{cla:bigtheta_control1}. Three cases from top to bottom: $x_{t(i)}=b_i$, $x_{t(i)}=b_j$ ($j\in\interval{i+1}{N+i-1}$), $x_{t(i)}=a_j$. The dashed blue lines represent the pairing. The edges $e_i$ and $e_{i+N}$ are not drawn, to emphasize the fact that they need not be covered since they contain $c_i$.}\label{fig:claim_intervals}
	\end{figure}
	
	\begin{claim}\label{cla:bigtheta_control2}
		There exist $(t,i) \in \interval{0}{T-1} \times \interval{1}{N}$ such that $c_i \in M_t$ and $c_i \not\in M_{t+1}$. In other words, at some point during the game, Maker removes a token from some $c_i$. Moreover, for such $(t,i)$ with $t$ minimal, we have $M_t \cap I_i^L = \varnothing$ or $M_t \cap I_i^R = \varnothing$.
	\end{claim}
	\begin{proofclaim}[Proof of Claim \ref{cla:bigtheta_control2}]
		Suppose for a contradiction that the first assertion is false: then $c_i \in M_T$ for all $i \in \interval{1}{N}$ by Claim \ref{cla:bigtheta_control1}. Since $e_N \subseteq M_T$ by Claim \ref{cla:bigtheta_last}, we get $|M_T| \geq N+3$, contradicting the fact that Maker only has $N+1$ tokens.
		
		As for the second assertion, let $(t,i) \in \interval{0}{T-1} \times \interval{1}{N}$ such that $c_i \in M_t$ and $c_i \not\in M_{t+1}$, with $t$ minimal. Suppose for a contradiction that there exist $v^L \in M_t \cap I_i^L$ and $v^R \in M_t \cap I_i^R$. Note that $v^L \neq v^R$: indeed, $I_i^L \cap I_i^R=\{b_N\}$ by definition, and Claim \ref{cla:bigtheta_last} ensures that $b_N \not\in M_t$ since $t<T$. Now, for all $j \in \interval{1}{N}$:
		\begin{itemize}[noitemsep,nolistsep]
			\item We have $t \geq t(j)$. Indeed, if $j \leq i$ then $I_j \supseteq I_j^L \supseteq I_i^L \ni v^L$, and if $j \geq i$ then $I_j \supseteq I_j^R \supseteq I_i^R \ni v^R$.
			\item We know $c_j \in M_{t(j)-1}$ by Claim \ref{cla:bigtheta_control1}. By minimality of $t$, this token on $c_j$ is not removed before Round $t$, hence $c_j \in M_t$.
		\end{itemize}
		In conclusion, we have $\{c_1,\ldots,c_N,v^L,v^R\} \subseteq M_t$ hence $|M_t| \geq N+2$, again contradicting the fact that Maker only has $N+1$ tokens.
	\end{proofclaim}
	
    Let $(t,i) \in \interval{0}{T-1} \times \interval{1}{N}$ satisfying Claim \ref{cla:bigtheta_control2}, with $t$ minimal, and suppose that $M_t \cap I_i^L = \varnothing$ (the case $M_t \cap I_i^R = \varnothing$ is analogous). We are going to build a pairing $\Pi'$ that is complete in $\pos_t$, using several facts:
	\begin{itemize}[noitemsep,nolistsep]
		\item Recall that $\pos_t$ is the position obtained at the end of Round $t$ i.e. just after the token on $c_i$ has been removed, so $\pos_t$ contains Maker tokens on exactly $M_t \setminus \{c_i\}$. In particular, there is no Maker token inside $I_i^L \cup \{c_i\}$: all of these vertices may be used in our pairing.
		\item Since $c_i \in M_t$, there exists $t' \in \interval{1}{t}$ such that $x_{t'}=c_i$. This implies $y_{t'}=\ov{c_i}$ by definition of Breaker's strategy. As a result, all edges in $E_{c_i}$ have already been taken care of by Breaker in $\pos_{t'}$, and in $\pos_t$ as well since $t \geq t'$.
	\end{itemize}
	Define $\Pi'$ to be the same as $\Pi$ except that $\{c_i,\ov{c_i}\}$ and the pairs inside $I_i^L$ are replaced by $\{c_i,b_i\}$ and $\{a_{\ell},b_{\ell+1}\}$ for $\ell \in \interval{i}{N-1}$, as in Figure \ref{fig:claim_intervals_bis}.
	
	\begin{figure}[h]
	\centering
	\includegraphics[scale=.58]{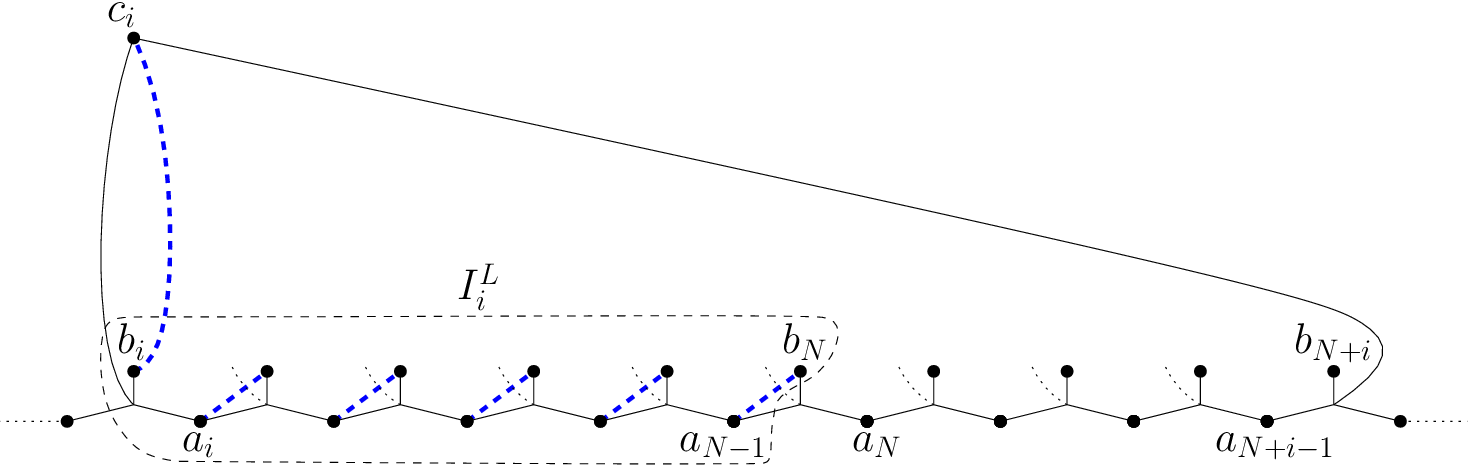}
	\caption{The pairing used inside $I_i^L \cup \{c_i\}$ if $M_t \cap I_i^L = \varnothing$. The dashed blue lines represent the pairing.}\label{fig:claim_intervals_bis}
	\end{figure}

	These new pairs cover $e_i,\ldots,e_N$, so the above facts ensure that $\Pi'$ is complete in $\pos_t$. In particular, if $y$ denotes the twin of $x_{t+1}$ in $\Pi'$ (or an arbitrary vertex if $x_{t+1}$ is in no pair of $\Pi'$), then $\Pi'$ is also complete in $(\pos_t,\{x_{t+1}\},\{y\})$, which contradicts Claim \ref{cla:bigtheta_pairing}.
\end{proof}

\section{Complexity results}\label{section5}

\subsection{General results}

Since the $(*,*)$-game coincides with the classical Maker-Breaker game, token positional games on $4$-uniform hypergraphs inherit \PSPACE-hardness from the classical version.

\begin{theorem}[{\cite[Theorem 3.1]{galliot2025}}]
    Deciding the winner of a token positional game is \PSPACE-hard, even restricted to $4$-uniform hypergraphs.
\end{theorem}

However, contrary to classical Maker-Breaker games, if there are strictly less tokens than vertices (when adding up both players' tokens) then one cannot ensure that the game will end in a polynomial number of moves. Hence, token positional games might not be in \PSPACE. However, we can prove that they lie in \EXPTIME.

\begin{proposition}
    Deciding the winner of a token positional game can be done in \EXPTIME.
\end{proposition}

\begin{proof}
    Let $\hyp$ be a hypergraph, and let $n = |V(\hyp)|$. Each vertex of $\hyp$ is either free, has a Maker token on it, or has a Breaker token on it. Therefore, the number of different positions of the game is $3^n$. If Maker has a winning strategy, then she has one in which the same position never appears twice, so the game ends in at most $3^n$ moves. Therefore, the minimax algorithm can explore all possible positions in \EXPTIME. 
\end{proof}

Let us also recall that, by Theorem \ref{theo:breaker1}, token positional games where Breaker only has one token can be solved in polynomial time.

\subsection{{\sf XP} algorithm}

The idea for the {\sf XP} algorithm is that, when both players have a bounded number of tokens, the set of all positions has polynomial size and can thus be fully explored.

In what follows, a {\em game state} refers to a position of a token positional game plus the knowledge of which player is next to play. Given a hypergraph $\hyp$ and two integers $a,b \in \Zp$, we define the directed graph $G_{a,b}(\hyp)$ whose vertex set is the set of all possible game states for the $(a,b)$-game on $\hyp$, and where there is an arc $(u,v)$ if and only if the game state $v$ can be reached from the game state $u$ in exactly one move (played by the player prescribed by $u$).

\begin{lemma}\label{lem:xp}
    Let $\hyp$ be a hypergraph, and let $a,b \in \Zp$. The graph $G_{a,b}(\hyp)$ has at most $2n^{2k}$ vertices and $2kn^{2k+1}$ arcs, where $n=|V(\hyp)|$ and $k=\max(a,b)$.
\end{lemma}

\begin{proof}
    A game state is described by three choices: the set of vertices occupied by Maker, the set of vertices occupied by Breaker, and the identity of the next player. This amounts to at most $\binom{n}{k} \times \binom{n}{k} \times 2 $ possibilities. Since $\binom{n}{k} \leq n^k$, the number of vertices of $G_{a,b}(\hyp)$ is at most $2n^{2k}$.

    To bound the number of arcs, we bound the maximum out-degree of a vertex in $G_{a,b}(\hyp)$. From any game state $u$, a move consists in selecting a token (at most $k$ choices) and its destination vertex (at most $n$ choices) so the out-degree of $u$ is at most $kn$. Therefore, the number of arcs of $G_{a,b}(\hyp)$ is at most $2kn^{2k+1}$.
\end{proof}

We conclude by encoding our token positional game as a {\em reachability game}. Such a game is played on a directed graph $G$, with a given subset $T$ of vertices which we see as a {\em target set}. A unique token is initially placed on a prescribed vertex of $G$. The players then take turns moving this token along an arc of $G$, i.e. the new vertex must be an out-neighbor of the previous one. The first player wins if the token ever sits on a vertex in $T$. If, instead, the game lasts indefinitely without any vertex in $T$ being visited, then the second player wins. It is known that deciding the winner of a reachability game can be done in linear time, by computing the 0-attractor of the target set \cite{reachability1,reachability2}.

\begin{theorem}
    Deciding the winner of a token positional game is in {\sf XP} parameterized by the number of tokens of both players. More precisely,
    deciding the winner of the $(a,b)$-game on a hypergraph on $n$ vertices can be done in time $O(n^{2\max(a,b)+2})$.
\end{theorem}

\begin{proof}
    Let $F$ be the set of all game states of the $(a,b)$-game on $\hyp$ in which some edge of $\hyp$ is filled with Maker tokens. It is straightforward that whichever player has a winning strategy for the reachability game on $G_{a,b}(\hyp)$ with target set $F$ also has a winning strategy for the $(a,b)$-game on $\hyp$. Note that our rule that Breaker wins if the same game state is reached twice is not an issue here, since no vertex is visited twice when the first player carries out a winning strategy for the reachability game (we would get an infinite loop of moves, contradicting the fact that $T$ is eventually reached).

    Therefore, using the linear-time algorithm from \cite{reachability2}, deciding the winner of the $(a,b)$-game on $\hyp$ can be done in time $O(|V|+|A|)$, where $V$ and $A$ are the vertex set and the arc set of $G_{a,b}(\hyp)$ respectively. By Lemma \ref{lem:xp}, defining $k=\max(a,b)$, we get a running time $O(n^{2k} + kn^{2k+1}) = O(n^{2k+2})$.
\end{proof}

\subsection{\EXPTIME-completeness for token sliding}

We study here the algorithmic complexity of the {\em token sliding} version of token positional games. This means that a token can only be moved from a vertex $u$ to a vertex $v$ if there is an edge of the hypergraph containing both $u$ and $v$, as opposed to the constraint-free {\em token jumping} version which was studied in all previous sections. The study of this variation is motivated by analogy with reconfiguration problems, in which both token jumping and token sliding settings are typically studied~\cite{Demaine2014,Bartier2021}. Note that all edges are relevant in the token sliding version, including those of size larger than Maker's number of tokens, because they might allow certain moves to be played. Also remark that, up to adding an edge containing all the vertices of the hypergraph (which cannot be filled), it is always possible to emulate token jumping in the token sliding setting.

We will perform a reduction from the {\em eternal dominating set} problem, which is defined as follows. Let $G$ be a graph, and let $D$ be a subset of vertices which we call {\em guards}. The eternal domination game is played by two players, Attacker and Defender. In each turn, Attacker selects ({\em attacks}) a vertex $v \not\in D$, then Defender chooses some guard $u \in D$ that is a neighbor of $v$ and moves it to $v$, and the game continues with the new set of guards $(D \cup \{v\}) \setminus \{u\}$. If, at any point, Defender has no legal move i.e. no guard is a neighbor of $v$, then Attacker wins. Otherwise, Defender wins i.e. the game lasts indefinitely. When Defender has a winning strategy, we say that the initial set $D$ is an {\em eternal dominating set} of $G$. Virgile proved that determining whether $D$ is an eternal dominating set of $G$ is \EXPTIME-complete~\cite{virgelot2024}. We use it to prove our result, which addresses general positions of the token sliding game, meaning that we consider games that start from a position in which some tokens may already be placed on the board.

\begin{theorem}\label{thm:exptime-token-sliding}
    Deciding the winner of a token positional game with token sliding is \EXPTIME-complete, even restricted to hypergraphs of rank 2 and Maker having a single token.
\end{theorem}

\begin{proof}
Let $(G, D)$ be an instance of the eternal dominating set problem, where we write the vertex set of $G$ as $\{ v_1, \ldots, v_n \}$.
We build an instance position $\pos = (\hyp, M, B)$ as follows. Since this proof involves graphs and hypergraphs, we will exceptionally use the word ``hyperedge'', reserving the word ``edge'' for graphs.
\begin{itemize}
\item $V(\hyp) = \{ u_1, \ldots, u_n \} \cup \{ u'_1, \ldots, u'_n \}$.
\item $E(\hyp) = E_1 \cup E_2 \cup E_3 \cup E_4$, where:
\begin{itemize}
    \item $E_1 = \{\{u'_i\} \mid i \in \interval{1}{n} \}$;
    \item $E_2 = \{\{u_i,u'_i\} \mid i \in \interval{1}{n} \}$;
    \item $E_3 = \{\{u_i,u_j\} \mid i \neq j \in \interval{1}{n} \}$;
    \item $E_4 = \{\{u'_i,u'_j\} \mid i \neq j \in \interval{1}{n}, \text{$v_i$ and $v_j$ are adjacent in $G$} \}$.
\end{itemize}
\item $M = \{ u_{i_0} \}$ for some arbitrary $i_0 \in \interval{1}{n}$ such that 
$v_{i_0} \in D$.
\item $B = \{ u'_i \mid i \in \interval{1}{n}, v_i \in D \}$.
\end{itemize}

This construction is depicted in Figure~\ref{fig:construction-exptime-token-sliding}. Note that, since Maker only has one token, she cannot win through the hyperedges of size $2$: those are only here to allow for sliding.

Let us prove that Attacker wins the eternal domination game on $(G, D)$ if and only if Maker wins the $(1,|D|)$-game with token sliding starting from the position $\pos$. For both directions of this equivalence, the idea is to mimic the winning strategy from the eternal domination game into the token positional game in such a way that, at the start of every round, Breaker's tokens are on the $u'_j$'s corresponding to the $v_j$'s that host Defender's guards in the parallel eternal domination game (note that, by definition of $B$, this is the case at the start of the first round).

First, suppose that Attacker has a winning strategy $\strat$ for the eternal domination game on $(G, D)$. For this, anytime $\strat$ prescribes to attack some $v_i$, Maker simply slides her token to $u_i$ (as permitted by the hyperedges in $E_3$). Since Defender has no guard on $v_i$ at that point, we know Breaker has no token on $u'_i$, so Maker is threatening to win on her next move by sliding her token from $u_i$ to $u'_i$ (as permitted by the hyperedges in $E_2$). As such, there are two cases. If Breaker fails to slide some token to $u'_i$ himself immediately, then he loses. Otherwise, Breaker slides a token from some $u'_k$ to $u'_i$, and we consider that Defender has moved a guard from $v_k$ to $v_i$ in the eternal domination game (as permitted by the fact that the hyperedges in $E_4$ correspond to the edges of $G$). Since $\strat$ is a winning strategy for Attacker, we will be in the former case at some point, meaning Maker will win.

Now, suppose that Defender has a winning strategy $\strat$ for the eternal domination game on $(G, D)$.
It suffices to show that Breaker can ensure that Maker's token remains on one of the $u_j$'s at all times. Note that, by definition of $M$ and $B$, this is the case after Maker's very first move: indeed, Breaker has a token on $u'_{i_0}$ at the start of the game, which prevents Maker from sliding her token from $u_{i_0}$ to $u'_{i_0}$. Now, anytime Maker moves her token to some $u_i$, there are two cases. If Breaker already has a token on $u'_i$, then he passes his turn. Otherwise, we consider that Attacker attacks the vertex $v_i$: the strategy $\strat$ then prescribes to move a guard from some $v_k$ to $v_i$, which Breaker replicates by sliding a token from $u'_k$ to $u'_i$ (as permitted by the fact that the hyperedges in $E_4$ correspond to the edges of $G$). Since $\strat$ is a winning strategy for Defender, Breaker can carry out this strategy indefinitely. Consequently, at the end of each round, Maker's token is on some vertex $u_i$ such that Breaker has a token on $u'_i$, so Maker can never win.
\end{proof}

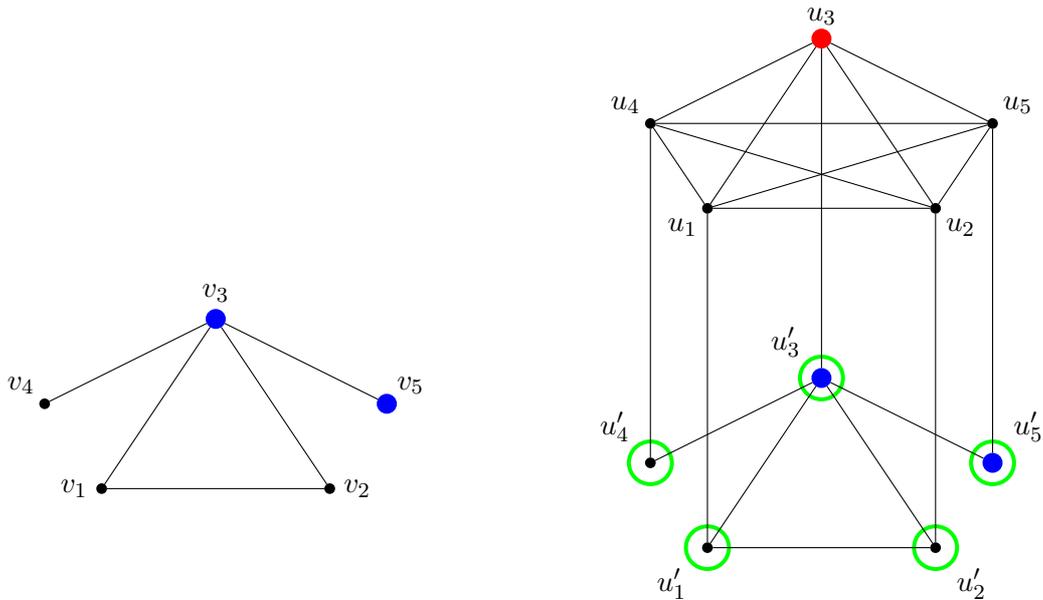
\begin{figure}[h]
    \centering

\begin{subfigure}{.4\linewidth}

\centering

\begin{tikzpicture}[scale=1.5]

\draw (0,0) node[v] (v1){} node[left = .05] {$v_{1}$};
\draw (2,0) node[v] (v2){} node[right = .05] {$v_{2}$};
\draw (1,1.5) node[B] (v3){} node[above = .1] {$v_{3}$};
\draw (-.5,.75) node[v] (v4){} node[above left] {$v_{4}$};
\draw (2.5,.75) node[B] (v5){} node[above right] {$v_{5}$};

\draw (2.5,-1) node[inv] {};

\draw (v1) -- (v2) -- (v3) -- (v1); 
 \draw (v4) -- (v3) -- (v5) ;

\end{tikzpicture}

\caption{An instance $(G, D)$ of the eternal domination game. The set of guards $D$ is represented by blue vertices.}
    
\end{subfigure} \hfil \begin{subfigure}{.45\textwidth}

\centering

\begin{tikzpicture}[scale=1.5]
\draw (0,0) node[rond]{};
\draw (2,0) node[rond]{} ;
\draw (1,1.5) node[rond]{} ;
\draw (-.5,.75) node[rond]{};
\draw (2.5,.75) node[rond]{} ;

\draw (0,0) node[v] (v1){} node[below left = 2mm] {$u'_{1}$};
\draw (2,0) node[v] (v2){} node[below right = 2mm] {$u'_{2}$};
\draw (1,1.5) node[B] (v3){} node[above left = 2mm] {$u'_{3}$};
\draw (-.5,.75) node[v] (v4){} node[above left = 2mm] {$u'_{4}$};
\draw (2.5,.75) node[B] (v5){} node[above right = 2mm] {$u'_{5}$};

\draw (v1) -- (v2) -- (v3) -- (v1); 
 \draw (v4) -- (v3) -- (v5) ;

\draw (0,3) node[v] (u1){} node[below left] {$u_{1}$};
\draw (2,3) node[v] (u2){} node[below right] {$u_{2}$};
\draw (1,4.5) node[R] (u3){} node[above = .05] {$u_{3}$};
\draw (-.5,3.75) node[v] (u4){} node[above left] {$u_{4}$};
\draw (2.5,3.75) node[v] (u5){} node[above right] {$u_{5}$};

\draw (u1) -- (u2) -- (u3) -- (u4) -- (u5) -- (u1); 
 \draw (u1) -- (u3) -- (u5) -- (u2) -- (u4) -- (u1) ; 

\draw (u1) -- (v1);
\draw (u2) -- (v2);
\draw (u3) -- (v3);
\draw (u4) -- (v4);
\draw (u5) -- (v5);

\end{tikzpicture}

\caption{The corresponding hypergraph from our reduction. The red vertex represents Maker's token; the blue vertices represent Breaker's tokens.}
    
\end{subfigure}
    
    \caption{The construction from the proof of Theorem~\ref{thm:exptime-token-sliding}.  Hyperedges of size $1$ are represented by green circles, while hyperedges of size $2$ are represented by black lines.}
    \label{fig:construction-exptime-token-sliding}
\end{figure}

\section{Conclusion}\label{section6}

In this paper, we have initiated the study of Maker Breaker token positional games. In particular, we have shown that the token number $\theta(\hyp)$ of a $k$-uniform hypergraph $\hyp$ (if finite) is equal to $k$ for $k \in \{2,3\}$, but can vary from $k$ to $\Omega(|V(\hyp)|)$ for all $k \geq 4$. Our results also give a better understanding of classical Maker-Breaker games, in two ways. First, since a high token number means that Maker's moves can still be useful long after they are played, we can see that Maker's winning strategies can become much more complex for $k=4$ compared to $k=3$. This is consistent with the gap in algorithmic complexity: classical Maker-Breaker games are tractable for $k=3$ \cite{GGS25} but \PSPACE-complete for $k=4$ \cite{galliot2025}. Second, we have shown that, for all $k \geq 4$, there are arbitrarily large hypergraphs $\hyp$ on which the duration $\tau(\hyp)$ of the classical Maker-Breaker game attains the trivial upper bound $\lceil |V(H)|/2 \rceil$, again in contrast with the case $k=3$. It remains an open question whether the token number also attains this bound, since the highest token number we have managed to construct is around $|V(H)|/6$.

Instances where Breaker has a single token ($b=1$) are completely understood, and we have provided a polynomial-time algorithm to solve them. One may hope that the concept of reducible pairs of edges, which is key to the case $b=1$, would admit some generalization to any constant value of $b$ and maybe yield an {\sf FPT} algorithm parameterized by $b$. However, this does not seem to be the case. Indeed, even though Maker trivially needs two simultaneous threats to win when $b=1$, she may not need three simultaneous threats to win when $b=2$ for instance. As an illustration of this, consider a linear 3-uniform cycle (i.e. a necklace without the initial token on it) of length at least 4: it can be shown that Maker wins the $(3,2)$-game, despite never threatening a one-move win in three different edges. A structural characterization of hypergraphs on which Maker wins the $(3,2)$-game would be a good step towards understanding the case $b=2$.

Our \EXPTIME-completeness result for the token sliding version of the game has room for improvement. As it addresses general positions, with tokens already sitting on the board at specific locations, the complexity remains unknown for starting positions of the game. Moreover, we would welcome an \EXPTIME-completeness result for ``standard'' token positional games, i.e. the token jumping version.

Finally, conventions other than Maker-Breaker may be considered for token positional games. In particular, it would be natural to study Maker-Maker token positional games, where whichever player first fills an edge wins.

\bibliographystyle{plain}
\bibliography{references}

\end{document}